\documentclass[12pt]{amsart}

\usepackage{lscape,amssymb, amsmath,verbatim,graphicx}
\usepackage{epsfig,float}
\usepackage{graphicx}
\usepackage{epsfig}
\usepackage{booktabs}
\newcommand{\lf}{\lfloor}
\newcommand{\cov}{\mbox{cov}}
\newcommand{\rf}{\rfloor}

\newcommand{\bdelta}{{\mbox{\boldmath $\psi$}}}

\newcommand{\balpha}{{\mbox{\boldmath $\delta$}}}

\newcommand{\bgamma}{{\mbox{\boldmath $\gamma$}}}
\newcommand{\bGamma}{{\mbox{\boldmath $\Gamma$}}}
\newcommand{\bLambda}{{\mbox{\boldmath $\Lambda$}}}
\newcommand{\bO}{{\bf O}}
\newcommand\T{\top}
\newcommand{\bX}{{\bf X}}
\newcommand{\bZ}{{\bf Z}}
\newcommand{\be}{{\bf e}}
\newcommand{\bH}{{\bf H}}

\newcommand{\G}{{\bf G}}
\newcommand{\U}{{\bf U}}
\newcommand{\V}{{\bf V}}
\newcommand{\br}{{\bf r}}
\newcommand{\vare}{{\nu}}
\newcommand{\A}{{\bf A}}
\newcommand{\fe}{{\mathfrak e}}
\newcommand{\hfe}{\hat{{\mathfrak e}}}
\newcommand{\C}{{\bf C}}
\newcommand{\hC}{\hat{\bf C}}
\newcommand{\tC}{\tilde{\bf C}}

\newcommand{\sa}{\begin{align*}
}
\newcommand{\se}{\end{align*}
}
\numberwithin{table}{section}
\numberwithin{figure}{section}

\newcommand{\bG}{{\bf G}}
\newcommand{\bv}{{\bf v}}
\newcommand{\W}{{\bf W}}
\newcommand{\baX}{\bar{\bf X}}

\newtheorem{rem}{Remark}[section]
\newtheorem{cor}{Corollary}[section]
\newtheorem{theorem}{Theorem}[section]
\newtheorem{lemma}{Lemma}[section]
\newtheorem{assumption}{Assumption}[section]
\usepackage{subfigure}

\usepackage[papersize={8.5in,11in},left=1.25in,right=1.25in,top=1in,bottom=1in]{geometry}

\newcommand\X{{\bf X}}

\newtheorem{definition}{Definition}[section]

\allowdisplaybreaks
\def\beq{\begin{equation}}
\def\eeq{\end{equation}}

\numberwithin{equation}{section}
\numberwithin{theorem}{section}

\begin{document}

\title[Eigenvalue Based Testing for Structural Breaks]{Empirical eigenvalue based testing for structural breaks in linear panel data models}

\author {Lajos Horv\'ath}

\address{ Department of Mathematics, University of Utah, Salt Lake City, UT, USA
}

\author{Gregory Rice}
\address{ Department of Statistics and Actuarial Science, University of Waterloo, Waterloo, ON, Canada}


\keywords{panel data, change point detection, time series, empirical eigenvalues, CUSUM process, weak convergence}


\begin{abstract}
Testing for stability in linear panel data models has become an important topic in both the statistics and econometrics research communities. The available methodologies address testing for changes in the mean/linear trend, or testing for breaks in the covariance structure by checking for the constancy of common factor loadings. In such cases when an external shock induces a change to the stochastic structure of panel data, it is unclear whether the change would be reflected in the mean, the covariance structure, or both. In this paper, we develop a test for structural stability of linear panel data models that is based on monitoring for changes in the largest eigenvalue of the sample covariance matrix. The asymptotic distribution of the proposed test statistic is established under the null hypothesis that the mean and covariance structure of the panel data's cross sectional units remain stable during the observation period. We show that the test is consistent assuming common breaks in the mean or factor loadings. These results are investigated by means of a Monte Carlo simulation study, and their usefulness is demonstrated with an application to U.S. treasury yield curve data, in which some interesting features of the 2007-2008 subprime crisis are illuminated.
\end{abstract}

\maketitle

\section{Introduction }\label{first}


We consider in this paper the problem of testing for the presence of a structural break in linear panel data models. Structural breaks in panel data may result from any of a number of sources. For example, if the data under consideration consists of U.S. macroeconomic indicators, then the onset of a recession, or the introduction of a new technology, may be evidenced by changes in the correlations between indicators or linear model parameters fitted from the data.



Change point analysis has been extensively developed to study such features in data; we refer to Aue and Horv\'ath (2012) for a recent survey of the field in the context of time series. Adapting change point methodology to the panel data setting presents a difficulty since the dimension, or number of cross sectional units ($N$), may be larger in relation to the sample size ($T$) than is typical in classical change point analysis. This encourages asymptotic frameworks in which both $N$ and $T$ tend jointly to infinity.

Most of the literature in this direction address either testing for changes in the mean, or testing for changes in the correlation structure as measured by changes in common factor loadings. With regards to testing for and estimating changes in the mean, we refer to Bai (2010), who derives a least squares change point estimator. Kim (2011, 2014), and Baltagi et al. (2015) extend this methodology to account for changes in linear trends in the presence of cross sectional dependence modeled by common factors. Horv\'ath and Hu\v{s}kov\'a (2012) develop a test for a structural change in the mean based on the CUSUM estimator. Li et al. (2014) and Qian and Su (2014) consider multiple structural breaks in panel data, and Kao et al. (2014) considers break testing under cointegration.


 Estimating and testing for changes in the covariance of scalar and vector valued time series of a fixed dimension are considered in Galeano and Pe\~na (2007), Aue et al.\ (2009), and Wied et al\ (2012).  With regards to testing for changes in the factor structure of panel data, Breitung and Eickmeier (2011) develop methodology that relies on testing for constancy of the least squares estimates obtained by regression on the principal component factors. Their test depends on estimating the number of common factors according to the information criterion developed in Bai and Ng (2002). In both the testing procedure, and the method used to determine the number of common factors, it is presumed that the mean remains constant.

In such instances when external shocks induce a change to the stochastic structure of panel data, it is unclear whether or not the change would affect the mean, the covariance structure, or both. Methods for detecting changes in the mean appear to be somewhat robust to small changes in the covariance structure of the panels, however the methods proposed in Breitung and Eickmeier (2011) to test for changes in the common factor loadings are sensitive to both changes in the mean, and large changes in the covariance, evidenced by non-monotonic power. This was recently addressed in Yamamoto and Tanaka (2015), in which a correction is proposed, but it raises the question of whether alternative methods to estimating principal components, and the number of common factors, might be effective in terms of detecting instability in panel data.

The alternative that we explore here relies on analyzing the largest eigenvalues of the covariance matrix. Using the largest eigenvalues of a covariance matrix as a simplified summary of the covariance structure of multivariate time series has served an important role in finance and econometrics for quite some time. This idea is utilized in Markowitz portfolio optimization (cf.\ Markowitz (1952, 1956)), and to model co--movements of markets and stocks as a barometer for risk (cf. Keogh et al.\ (2004) and Zovko and Farmer (2007)), among other applications.

In this paper, we propose methodology for testing structural stability in linear panel data models that is based on a process derived from the largest eigenvalue of the covariance matrix based on an increasing proportion of the total sample.  The asymptotic distribution of the eigenvalue process is established assuming structural stability. Furthermore, we show that functionals of the eigenvalue process diverge when there is a common break in the mean or covariance as measured by the common factor loadings.


The rest of the paper is organized as follows. In Section \ref{main-1}, we present the linear panel data models and assumptions considered in the paper, as well as the main asymptotic results for the largest eigenvalue under the null hypothesis of stability of the model parameters. Section \ref{alt} contains the details of applying the results of Section \ref{main-1} to the change point problem, including asymptotic consistency results under the mean break and factor loading break alternatives. In Section \ref{monte}, we discuss the practical implementation of the test, and present the results of a Monte Carlo simulation study. Section \ref{app} contains an application of the methodology developed in the paper to US treasury yield curve data. Analogous results for smaller eigenvalues are considered in Section \ref{all-eig}. All proofs of the technical results are collected in Section \ref{proofs}.

\section{Models, assumptions, and asymptotics under $H_0$}\label{main-1}

We consider the model

\beq\label{model-full}
X_{i,t}=(\mu_i+\delta_iI\{t\geq t^*\})+(\gamma_i+\psi_iI\{t\geq t^*\})\eta_t+e_{i,t},\;\;1\leq i \leq N, 1\leq t \leq T,
\eeq

where $X_{i,t}$ denotes the $i^{\mbox{th}}$ cross section of the panel at time $t$, $\mu_i$ denotes the initial mean of the $i^{\mbox{th}}$ cross section that changes to $\mu_i+\delta_i$ at the unknown time $t^*$, $\eta_t$ denotes a real valued common factor with initial loadings $\gamma_i$ that may change to $\gamma_i+\psi_i$, and $e_{i,t}$ denote the idiosyncratic errors. It is presumed that both the common factor and idiosyncratic errors may be serially correlated. As we develop asymptotics, we assume that the number of cross sections $N$ depends on the observation period $T$, and $N$ is allowed to tend to infinity with $T$. We make the assumption that $\eta_t \in {\mathbb R}$ for the sake of simplicity; these results could be extended to the more general case of a vector valued common factor and factor loading.

We are interested in testing the null hypothesis that the model parameters remain stable during the observation period $1\leq t \leq T$, i.e.
$$
H_0:\; t^*>T.
$$
When $H_0$ holds, the model of \eqref{model-full} reduces to
\beq\label{model-null}
X_{i,t}=\mu_i+\gamma_i\eta_t+e_{i,t},\;\;1\leq i \leq N, 1\leq t \leq T.
\eeq
Let $\cdot^\T$ denote the matrix transpose, and define the vectors $\X_t=(X_{1,t}, X_{2,t},\ldots ,X_{N,t})^\T\in {\mathbb R}^N$. We define

\beq\label{c-def}
\hC_{N,T}(u)=\frac{1}{\lf Tu \rf}\sum_{t=1}^{\lf Tu\rf}(\X_t-\baX_T)(\X_t-\baX_T)^\T,\;\;1/T\leq u\leq 1,
\eeq

to be the sample covariance matrix based on the proportion $u$ of the sample, where

$$
\baX_T=\frac{1}{T}\sum_{t=1}^T\X_t.
$$


In order to test $H_0$, we utilize the processes derived from the $K$ largest eigenvalues $\hat{\lambda}_{1}(u)\geq \hat{\lambda}_{2}(u)\geq\ldots \geq \hat{\lambda}_{K}(u)$ of $\hat{C}_{N,T}(u)$. We focus our attention at first on the process derived from the largest eigenvalue, and make the primary objective of this section is to establish the weak convergence of  $\hat{\lambda}_{1}(u)$ under $H_0$. Analogous results for processes derived from the smaller eigenvalues are provided in Section \ref{all-eig}. We note that an alternative to using $\hat{\lambda}_i(u)$ is to use $\tilde{\lambda}_i(u)=(\lf Tu\rf/T)\hat{\lambda}_i(u)$, which are equivalent with the largest eigenvalues of

\beq\label{c-deff}
\tC_{N,T}(u)=\frac{1}{ T}\sum_{t=1}^{\lf Tu\rf}(\X_t-\baX_T)(\X_t-\baX_T)^\T,\;\;0\leq u\leq 1.
\eeq

Assuming that $H_0$ holds, $\C=\mbox{cov}(\X_t)$ does not depend on $t$, and in this case we define the eigenvalues and eigenvectors of $\C$ by

\beq\label{ei-1}
\lambda_i\fe_i=\C\fe_i,\;\;1\leq i \leq N,
\eeq
where $\|\fe_i\|=1,\;\;1\le i \le N$, and $\| \cdot \|$ denotes the Euclidean norm in ${\mathbb R}^N$. Since $N$ is allowed to depend on $T$, both the eigenvalues $\lambda_i$ and eigenvectors $\fe_i$ may evolve as $T\to \infty$. Throughout this paper, we make use of the following assumptions:


\begin{assumption}\label{inc}
The eigenvalues $\lambda_1,\lambda_2,\ldots,\lambda_K$ satisfy that $\min_{1\leq i \leq K}(\lambda_{i}-\lambda_{i+1})\geq c_0$ for some constant $c_0>0$.
\end{assumption}

\begin{assumption}\label{b-g}
The common factor loadings satisfy that $|\gamma_i|\leq c_1\;\;\mbox{for all}\;\;1\leq i \leq N\;\;\mbox{with some}\;\; c_1>0.$
\end{assumption}

Assuming that the eigenvalues of $\C$ are distinct is necessary to derive a normal approximation for their estimates, and is a common assumption in the literature. We assume that the common factors and idiosyncratic errors satisfy a fairly general weak dependence condition.

\begin{definition}
{\rm We say that a stationary time series $\{\varepsilon_t,\; -\infty < t < \infty\}$ is an $L^p-m-$approximable Bernoulli shift with rate function $\chi$ if $E\varepsilon_t=0, \;\;E\varepsilon_t^{p}<\infty,$ and $\varepsilon_t=g(\vare_t, \vare_{t-1},\ldots)$ for some measurable function $g:{\mathbb R}^\infty \to {\mathbb R}$ where $\{\vare_s, -\infty<s<\infty\}$ are independent and identically distributed random variables, and $(E(\eta_t-\eta_t^{(m)})^{p})^{1/p}=\chi(m)$ with $\eta_t^{(m)}=g(\vare_t,\vare_{t-1}, \ldots, \vare_{t-m}, \vare^*_{t-m-1,t,m},\vare^*_{t-m-2,t,m},\ldots)$ and the $\vare^*_{i,j,\ell}$ are independent and identically distributed copies of $\vare_0$.}
\end{definition}

The space of stationary processes that may be represented as Bernoulli shifts is enormous; we refer to Wu (2005) for a discussion. Examples include stationary ARMA, ARCH, and GARCH processes. The rate function describes the rate at which such processes can be approximated with sequences exhibiting a finite range of dependence. In many examples of interest, the rate function may be taken to decay exponentially in the lag parameter.

\begin{assumption}\label{error-as}
{\rm
\begin{enumerate}
\item[(a)] $\{\eta_t,\;\; -\infty < t < \infty\}$ is $L^{12}-m-approximable$ with rate function $\chi_\eta(m) = c_2 m^{-\alpha_\eta}$ for constants $c_2>0$ and $\alpha_\eta>1$, and $E\eta_t^2=1$.
\item[(b)] The sequences $\{e_{i,t},\;\; -\infty < t < \infty\},\; 1 \le i \le N$, are each $L^{12}-m-approximable$ with rate functions $\chi_{e,i}(m) \le c_3 m^{-\alpha_e}$ for constants $c_3>0$ and $\alpha_e>1$. There exist constants $c_4$ and $c_5$ such that $0 < c_4 \le Ee^2_{i,_t}=\sigma_i^2\le c_5 <\infty$.
\item[(c)] The sequences $\{\eta_t,\;\; -\infty < t < \infty\},$ and $\{e_{i,t},\;\; -\infty < t < \infty\},\; 1 \le i \le N$ are independent.
\end{enumerate}
}

\end{assumption}

The least restrictive moment condition that could be assumed in order to obtain a normal approximation for the empirical eigenvalues is four moments. Our assumption of twelve moments comes from the fact that we apply a third order Taylor series expansion for the difference between the empirical eigenvalue process $\hat{\lambda}_i(u)$ and $\lambda_i$, (cf.\ Hall and Hosseini--Nasab (2009)) and twelve moments are needed to get an upper bound for the highest order term that is uniform with respect to $u$. The condition in Assumption \ref{error-as} that $E \eta_t^2 =1 $ is nonrestrictive; it makes the model \eqref{model-null} identifiable. In order to state the main result, we define


$$
\xi_{i,t} = \fe_i^\T (\bX_t - E \bX_0 ) (\bX_t - E \bX_0 )^\T \fe_i.
$$

\begin{theorem}\label{main} If $H_0$ and Assumptions \ref{inc}, \ref{b-g}, and \ref{error-as} hold, and

\beq\label{n-t-1}
\frac{N(\log T)^{1/3}}{T^{1/2}}\to 0, \;\;\mbox{as}\;\;T\to \infty,
\eeq
then
$$\frac{T^{1/2}}{\sigma_1}u(\hat{\lambda}_1(u)-\lambda_1)\stackrel{{\mathcal D[0,1]}}{\longrightarrow} W(u),$$
where $W(u)$ is a Wiener process, $\stackrel{{\mathcal D[0,1]}}{\longrightarrow}$ denotes weak convergence in the Skorokhod topology, and
$$
\sigma_1^2 =\sigma_1^2(T)= \sum_{t = -\infty}^\infty \mbox{{\rm cov}}(\xi_{1,0},\xi_{1,t}).
$$
\end{theorem}

Theorem \ref{main} shows that the distribution of the largest eigenvalue process may be approximated by a Brownian motion. We note that the norming sequence $\sigma_1^2$, which is essentially the long run variance of the quadratic forms $\xi_{1,t}$, may change with $N$. In fact, we show in Section \ref{proofs} that if $\bgamma=(\gamma_1,\gamma_2, \ldots, \gamma_N)^\T$, then under $H_0$ $\sigma_1^2 \to \infty$, as $T\to \infty$, if $\|\bgamma\| \to \infty$. The necessity of including the logarithm term in the rate condition \eqref{n-t-1} comes from the fact that we establish weak convergence on the entire unit interval. This condition can be improved by considering convergence on an interval that is bounded away from zero.

\begin{theorem}\label{main-2}
If the conditions of Theorem \ref{main} are satisfied and \eqref{n-t-1} is replaced with
\beq\label{n-t-2}
\frac{N}{T^{1/2}}\to 0,\;\;\;\mbox{as}\;\;T\to \infty,
\eeq
then for all $c \in (0,1]$,
$$\frac{T^{1/2}}{\sigma_1}u(\hat{\lambda}_1(u)-\lambda_1)\stackrel{{\mathcal D[c,1]}}{\longrightarrow} W(u).$$
where $\sigma_1^2$ is defined as in Theorem \ref{main}.
\end{theorem}

Conditions \eqref{n-t-1} and \eqref{n-t-2} require that the sample size $T$ is asymptotically larger than the squared dimension $N^2$. The case when $N$ is proportional to $T$ has received considerable attention in the probability and statistics literature. Assuming that $\hat{C}_{N,T}(1)$ is based on independent and identically distributed entries, the distribution of $\hat{\lambda}_1(1)$ converges to a Tracy-Widom distribution (cf. Johnstone (2008)). For a survey of the theory of eigenvalues of large random matrices, we refer to Aue and Paul (2014).

\section{Changepoint detection}\label{alt}

\subsection{Estimating the norming sequence}
Consistent estimation of $\sigma_1^2$ is required in order to apply Theorems \ref{main} and \ref{main-2} to test $H_0$. As $\sigma_1^2$ is defined as the long run covariance of the quadratic forms $\xi_{i,t}$, we propose a natural nonparametric estimator. We define $\hat{\fe}_i$ by
$$
\hat{\lambda}_i(1)\hat{\fe}_i=\hat{\C}_{N,T}(1)\hat{\fe}_i,\;\;1\leq i \leq N.
$$

Let $\hat{\xi}_{i,t}=(\hat{\fe}_i^\T(\X_t-\bar{\X}^*_{T,t}))^2,$ where

\begin{displaymath}
\bar{\X}^*_{T,t}=
\left\{
\begin{array}{ll}
\displaystyle \frac{1}{\hat{t}^*}\sum_{t=1}^{\hat{t}^*} \X_t ,\;\;\mbox{if}\;\;1 \le t \le \hat{t}^* \\
\displaystyle \frac{1}{T- \hat{t}^*}\sum_{t=\hat{t}^*+1}^{T} \X_t ,\;\;\mbox{if}\;\;\hat{t}^*+1 \le t \le T,
\end{array}
\right.
\end{displaymath}
and $\hat{t}^*$ is the least squares change point estimator for a change in the mean defined in Section 3 of Bai (2010). Estimating the mean under the alternative of a mean change is done to ensure monotonic power in that case. Let $J$ be a kernel/weight function that is continuous and symmetric about the origin in ${\mathbb R}$ with bounded support, and satisfying $J(0)=1$. Examples of such functions include the Bartlett and Parzen kernels; further examples and discussion may be found in Taniguchi and Kakizawa (2000). We define the estimator $\hat{v}^2_{1,T}$ for $\sigma_1^2$ by

\begin{align}\label{var-est-def}
\hat{v}^2_{1,T}=\sum_{s=-N+1}^{N-1}J\left(\frac{s}{h}\right)\hat{r}_{1,s},
\end{align}
where $h$ denotes a smoothing bandwidth parameter, and
\begin{displaymath}
\hat{r}_{1,s}=
\left\{
\begin{array}{ll}
\displaystyle \frac{1}{T-s}\sum_{t=1}^{T-s} (\hat{\xi}_{1,t}-\bar{\xi}_{1,T})(\hat{\xi}_{1,t+s}-\bar{\xi}_{1,T}),\;\;\mbox{if}\;\;s\geq 0
\vspace{.2cm}\\
\displaystyle\frac{1}{T-|s|}\sum_{t=-s}^{T} (\hat{\xi}_{1,t}-\bar{\xi}_{1,T})(\hat{\xi}_{1,t+s}-\bar{\xi}_{1,T}),\;\;\mbox{if}\;\;s< 0,
\end{array}
\right.
\end{displaymath}
with
$$
\bar{\xi}_{1,T}=\frac{1}{T}\sum_{t=1}^T\xi_{1,t}.
$$

\begin{theorem}\label{estivar} If $H_0$ and the conditions of Theorem \ref{main} are satisfied, and

\beq\label{h-1}
h=h(T)\to\infty \;\;\mbox{and}\;\;hN^3/T^{1/2}\to 0,\;\;\mbox{as}\;\;T\to \infty,
\eeq
then
\beq\label{cos-1-1}
\frac{\hat{v}^2_{1,T}}{\sigma_1^2}\;\;\stackrel{P}{\to}\;1, \;\;\;\mbox{as}\;\;T\to \infty.
\eeq
\end{theorem}


The results in Theorems \ref{main}, \ref{main-2}, and \ref{estivar} can be used to test for the stability of the largest eigenvalue, which, as we show below, suggests stability of the model parameters.

\begin{cor}\label{b-conv} Let
$$
\hat{B}_{T,1}(u)=\frac{T^{1/2}}{\hat{v}_{1,T}}u(\hat{\lambda}_1(u)-\hat{\lambda}_1(1)),\;\;0\leq u \leq 1.
$$
Under the conditions of Theorems \ref{main} and \ref{estivar}, $\hat{B}_{T,1}(u) \stackrel{{\mathcal D[0,1]}}{\longrightarrow} W^0(u)$, where $W^0$ is a standard Brownian bridge.

\end{cor}

The continuous mapping theorem and Corollary \ref{b-conv} imply that
\begin{align}\label{imp-form}
\sup_{0\le t \le 1} |\hat{B}_{T,1}(u)|  \stackrel{{\mathcal D}}{\to} \sup_{0 \le t \le 1} |W^0(t)|.
\end{align}

The limiting distribution on the right hand side of \eqref{imp-form} is commonly referred to as the Kolmogorov distribution. An approximate test of size $\alpha$ of $H_0$ is to reject if $\sup_{0\le t \le 1} |\hat{B}_{T,1}(u)| $ is larger than the $\alpha$ critical value of the Kolmogorov distribution. One could also consider alternate functionals of $\hat{B}_{T,1}$ to test $H_0$. The distributions of many functionals of $W^0$ are well--known (cf.\ Shorack and Wellner (1986), pp.\ 142--149).

\subsection{Consistency under alternatives}
We now turn our attention to studying the consistency of tests for $H_0$ based on $\sup_{0 \le t \le 1} |\hat{B}_{T,1}(u)|$ under the mean break and factor loading break alternatives. Following the literature, we assume that the change does not occur too close to the end points of the sample:

\beq\label{asth-1}
t^*=\lf T\theta\rf\;\;\mbox{with some }\;\;0<\theta<1.
\eeq

First we consider the case of a break in the mean, i.e. the model

\beq \label{mb-mod}
X_{i,t}=(\mu_i+\delta_iI\{t\geq t^*\})+\gamma_i\eta_t+e_{i,t},\;\;1\leq i \leq N, 1\leq t \leq T,
\eeq

holds. Let $\balpha=\balpha_T=(\delta_1, \delta_2, \ldots, \delta_N)^\T$ and assume
\beq\label{asth-2}
\lim_{T\to \infty}\frac{T^{1/2}\|\balpha\|}{\|\bgamma\|^2}=\infty.
\eeq


\begin{theorem}\label{th-cons} Under \eqref{mb-mod}, Assumptions \ref{inc}, \ref{b-g}, and \ref{error-as}, and assuming that \eqref{n-t-2}, \eqref{asth-1}, and \eqref{asth-2} are satisfied, then we have that
\beq\label{th-c-1}
\sup_{0\leq u\leq 1}|\hat{B}_{T,1}(u)|\;\;\stackrel{P}{\to}\;\;\infty
\eeq
\end{theorem}


We note that assumptions  \eqref{asth-1} and \eqref{asth-2} also appeared in Horv\'ath and Hu\v{s}kov\'a (2012) where  the optimality of these
conditions are discussed. It is clear if $N$ is large, relatively small changes can be detected by $\hat{\lambda}_1(u)$. As a consequence of the proof of Theorem \ref{th-cons}, it follows that

$$
\max_{2\leq i \leq K}\sup_{0\leq u\leq 1}T^{1/2}|\hat{\lambda}_i(u)-\hat{\lambda}_i(1)|=O_P(1),
$$
i.e.\ a change in the mean is asymptotically entirely captured by the largest eigenvalue of the partial covariance matrices.



The condition \eqref{asth-2} suggests how a local change in the mean alternative may be considered. For example, if $\delta_1=\delta_2=\ldots =\delta_N=\delta(N,T)$ and $\gamma_1=\gamma_2=\ldots =\gamma_N=\gamma$, $\gamma$ is fixed , we need that $(T/N)^{1/2}|\delta(N,T)|\to \infty$ for \eqref{asth-2} to hold, which describes at what rate $\delta(N,T)$ may tend to zero while maintaining consistency.\\

Next we consider the model
\beq\label{load-2}
X_{i,t}=\mu_i+(\gamma_i+\psi_iI\{t\geq t^*\})\eta_t+e_{i,t},\;\;1\leq i \leq N, 1\leq t \leq T,
\eeq
i.e.\ the means of the panels remain the same but the loadings change at time $t^*$. Let $\bdelta=(\psi_1,\psi_2, \ldots,\psi_N)^\T$.

\begin{theorem}\label{load-cons} Under \eqref{load-2},  Assumptions \ref{inc}, \ref{b-g}, and \ref{error-as}, and assuming that \eqref{n-t-2}, \eqref{asth-1} and
\beq\label{load-3}
\lim_{T\to \infty}\frac{(1-\theta)[\|\bdelta\|^2+2|\bdelta^\T\bgamma|] +(\bdelta^\T\bgamma)^2/\|\bdelta\|^2}{\|\bgamma\|^2+\max_{1\leq i \leq N}\sigma_i^2}>1
\eeq

hold, then $\sup_{0\leq u\leq 1}|\hat{B}_{T,1}(u)|\;\;\stackrel{P}{\to}\;\;\infty$, as $T \to \infty$.
\end{theorem}
 Roughly speaking, it is possible that the covariance might change on a subspace that is orthogonal to the first eigenvector (or more generally the first $K$ eigenvectors), and then if this change is not sufficiently large, the first eigenvalue cannot have power to detect it. Condition \eqref{load-3} is sufficient to imply that this does not occur.

\section{Finite Sample Performance}\label{monte}


In order to demonstrate how the result in \eqref{imp-form} is manifested in finite samples, we present here the results of a Monte Carlo simulation study involving several different data generating processes (DGP's) that follow \eqref{model-full}. All simulations were carried out in the R programming language (cf.\ R Development Core Team (2010)). In order to compute the long run variance estimate $\hat{{\it v}}^2_{1,T}$ defined in \eqref{var-est-def}, we used the ``sandwich" package (cf.\ Zeileis (2006)), in particular the ``kernHAC" function. The Parzen kernel with corresponding bandwidth defined in Andrews (1991) were employed.

\subsection{Empirical Size }


We begin by presenting the results on the empirical size of the test for stability based on the largest eigenvalue by considering two examples of synthetic data generated according to model \eqref{model-null}. We use the notation $Y_{i} \sim Y$ to denote that the sequence of random variables $Y_i$ are independent and identically distributed with distribution $Y$. Let $N_{i,t}(0,1)$ $i\ge0$ and $t \in \mathbb{Z}$ denote iid standard normal random variables, and let $AR_i(1,p)$ $i \ge 0$ denote independent autoregressive one processes with parameter $p$ based on standard normal errors. We generated observations $X_{i,t}$ according to \eqref{model-null} and the DGP's \\

(IID): $\eta_t= N_{0,t}(0,1),$ $e_{i,t} = s_i N_{i,t}(0,1)$, $s_i \sim Unif(.8,1.2)$, $\gamma_i \sim N(0,1)$,\\
and \\
(AR-1): $\eta_t=AR_0(1,.5),$ $e_{i,t} = s_i AR_i(1,.5)$,  $s_i \sim Unif(.8,1.2)$, $\gamma_i \sim N(0,1)$.\\

The purpose of choosing random parameters $s_i$, which define the standard deviations of the idiosyncratic errors, and $\gamma_i$ is two fold. Firstly, this forces Assumption \ref{inc} to hold. Secondly, this choice highlights that the methodology is relatively robust to variations in the parameter values.

Five simulated paths of the process $\hat{B}_{T,1}(u)$ are shown in the left hand panel of \ref{fig-1} when  $T=100$ and $N=20$, under IID. The most notable feature is that each process always starts with a spike near the origin, i.e. $\hat{\lambda}_i(u)$ is much larger than $\hat{\lambda}_i(1)$ when $u$ is small. The reason for this is that, when $u$ is small, $\hat{\lambda}_i(u)$ is computed from a matrix that is low rank, and hence will tend to be closer to the norm of the observation vectors, which is on the order of $N$, than the eigenvalue that it being estimated. This problem is ameliorated when $N$ decreases or $T$ increases, but significantly affects the results for many practical values of $N$ and $T$.

\begin{figure}
\centering
\caption{The left panel illustrates five simulated paths of $\hat{B}_{T,1}(u)$ when $N=20$ and $T=100$ under (IID), and the right panel illustrates five simulated paths of $ \tilde{B}_{T,1}(u)$ under the same conditions with $\epsilon=.05$.}\label{fig-1}
\mbox{\subfigure{\includegraphics[width=2.2in]{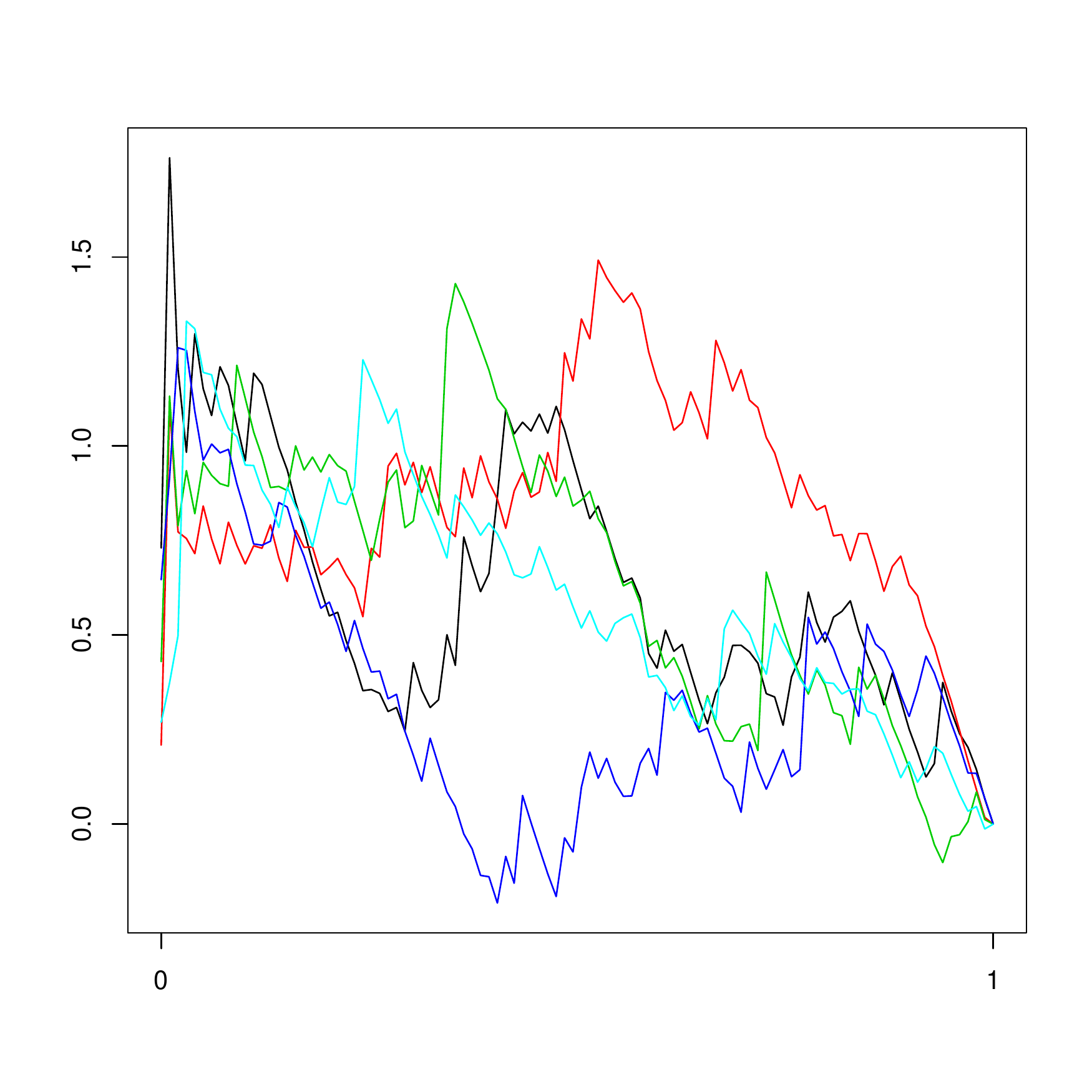}}\quad
\subfigure{\includegraphics[width=2.2in]{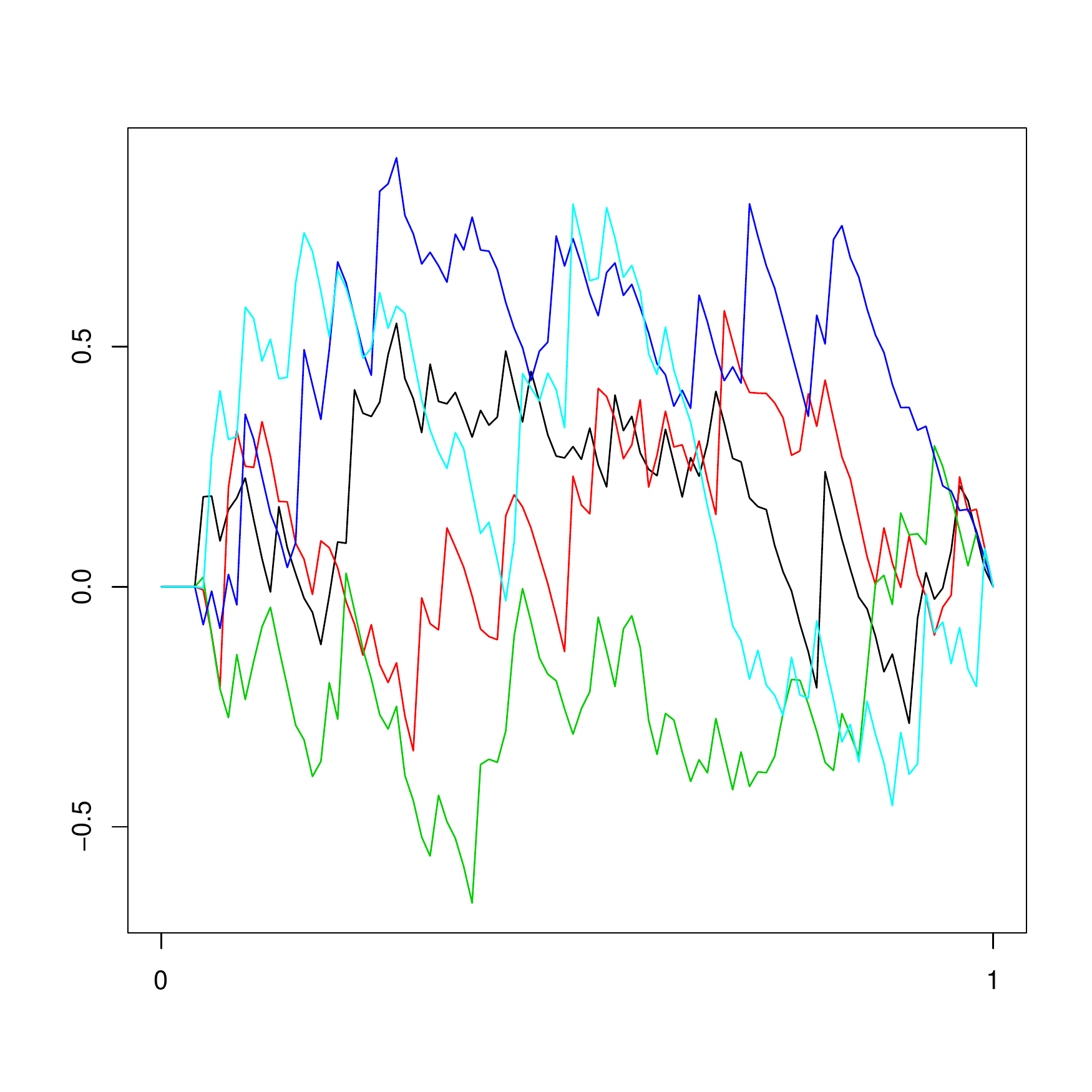} }}
\end{figure}

{\small
\begin{table}
\begin{tabular}{c c  ccc c  ccc c ccc c ccc c  }
\toprule
\multicolumn{2}{c}{DGP}&\multicolumn{7}{c}{IID} & &\multicolumn{7}{c}{AR-1}\\ \cmidrule{3-9} \cmidrule{11-17}
\multicolumn{2}{r}{} & \multicolumn{3}{c}{$\epsilon=.05$} & & \multicolumn{3}{c}{$\epsilon=.1$}& & \multicolumn{3}{c}{$\epsilon=.05$}& &  \multicolumn{3}{c}{$\epsilon=.1$}\\ \cmidrule{3-5}\cmidrule{7-9} \cmidrule{11-13} \cmidrule{15-17}
N&T & 10\% &5\%  &1\% & & 10\% &5\% & 1\% &  & 10\% &5\% &1\% & & 10\% &5\% & 1\% \\
\hline
10 & 50 &18.1&11.2&3.8 & &8.8&4.9&1.8 & &26.7&18.4&10.0 & &24.7&17.9&8.4 \\
 & 100 &8.3&3.5&.7 & &9.2&3.6&.7 & &17.1&10.3&3.4 & &9.2&3.6&.7 \\
 & 200 &8.7&4.1&.7 & &8.6&4.3&1.0 & &11.7&5.7&2.0 & &10.4&5.1&1.6 \\
\hline
20 & 50 &18.6&12.3&5.5 & &9.5&4.8&.7 & &23.7&16.5&8.0 & &25.8&17.8&8.7 \\
 & 100 &8.5&3.6&.6 & &9.1&4.5&.3 & &14.9&9.4&3.4 & &14.9&9.0&3.7 \\
 & 200 &8.4&4.2&.6 & &8.8&3.3&.5 & &11.8&6.5&2.0 & &12.4&6.8&1.5 \\
\hline
50 & 50 &23.3&13.7&5.3 & &10.2&3.9&.7 & &24.8&17.3&8.8 & &24.4&18.6&.9 \\
 & 100 &8.8&3.5&.6 & &9.0&4.2&1.0 & &17.8&11.6&4.0 & &15.3&8.5&3.5 \\
 & 200 &10.0&5.0&1.3 & &8.9&3.8&.5 & &13.0&7.1&2.1 & &12.2&6.4&1.7 \\
\bottomrule
\end{tabular}
\caption{Empirical sizes with nominal levels of 10\%, 5\%, and 1\% in both the independent (IID) and
dependent (AR-1) cases based on the process $\tilde{B}_{T,1}$.} \label{null-sim}
\end{table}
}

 In order to correct for this, we define

\begin{displaymath}
   \tilde{B}_{T,1}(u) = \left\{
     \begin{array}{lr}
       0 &  u \in [0,\epsilon]\\
       \vspace{.2cm} \\
       \hat{B}_{T,1}(u)-\frac{1-u}{1-\epsilon}\hat{B}_{T,1}(\epsilon) &  u \in (\epsilon,1]
     \end{array}
   \right.
\end{displaymath}

for a trimming parameter $\epsilon>0$. Five corresponding paths of $\tilde{B}_{T,1}(u)$ are illustrated in the right panel of Figure \ref{fig-1}, with $\epsilon=.05$.



Table \ref{null-sim} contains the percentages of the test statistic $\sup_{0\le u \le 1} |\tilde{B}_{T,1}(u)|$ that are larger than the 10\%, 5\%, and 1\% critical values of the Kolmogorov distribution. The results can be summarized as follows:

\begin{enumerate}
\item When $T$ is small ($T=50$), then the size of the test may be inflated by two sources. One of them is the spiked effect, and this is particularly pronounced when $\epsilon$ is small and $N$ is large. If the temporal dependence in the data is low, then increasing $\epsilon$ can allow the test to achieve good size even for small $T$ and relatively large $N$. However, strong temporal dependence can cause size inflation for small $T$ that cannot be accounted for by increasing $\epsilon$.

\item Another source of size inflation that is present for larger values of $T$ may be attributed to estimating the variance under the alternative of a break in the mean. This may be improved by considering alternative variance estimation approaches, such as those developed in Kejriwal (2009).
\item The difference in the results between the IID and AR-1 DGP's were small for larger values of $T$ $(T=100,200)$, indicating the variance estimation is performing well.
\item For $T=200$, the empirical sizes are close to nominal in all cases.
\end{enumerate}

\subsection{Empirical Power}

In order to study the power of our test under both the mean break and loading break alternatives, we considered two processes that satisfy \eqref{model-full} with $t^*=T\theta$ with $\theta \in (0,1)$. Throughout the simulations below, we set $t^*=T/2$, i.e. the break was in the middle of the sample. We also studied the situation in which breaks occured towards the endpoints of the sample. The results in those cases tended to be worse, but not more so than expected. We define the DGP's


MB($\delta$): $X_{i,t}=\delta_iI\{t\geq T/2\})+\gamma_i\eta_t+e_{i,t},\;\;1\leq i \leq N, 1\leq t \leq T$, where $\delta_i \sim \mbox{Unif}(-\delta,\delta)$\\
and \\
LB($\Delta$): $X_{i,t}=(\gamma_i+\psi_iI\{t\geq T/2\})\eta_t+e_{i,t},\;\;1\leq i \leq N, 1\leq t \leq T$, where $\psi_i \sim N(0,\Delta^2)$\\

In each case we take the other terms in \eqref{model-full}, i.e. the idiosyncratic errors, common factor, and factor loadings, to satisfy AR-1. We let the parameters $\delta$ and $\Delta$ vary between $0$ and $4$ at increments of $.5$, and let $N=10,20,50$, and $T=50,100,200$. The results are displayed in terms of power curves in Figures \ref{Mpower-Nf} and \ref{Mpower-Tf} in case of a mean break alternative (MB($\delta$)) and in Figures \ref{Lpower-Nf} and \ref{Lpower-Tf} in case of breaks in the factor loadings (LB($\Delta$)) when the size of the significance level of the test was fixed at 5\%. We summarize the results as follows:

\begin{figure}[h!]
\centering
      \includegraphics[width=0.6\textwidth,height=0.6\textwidth]{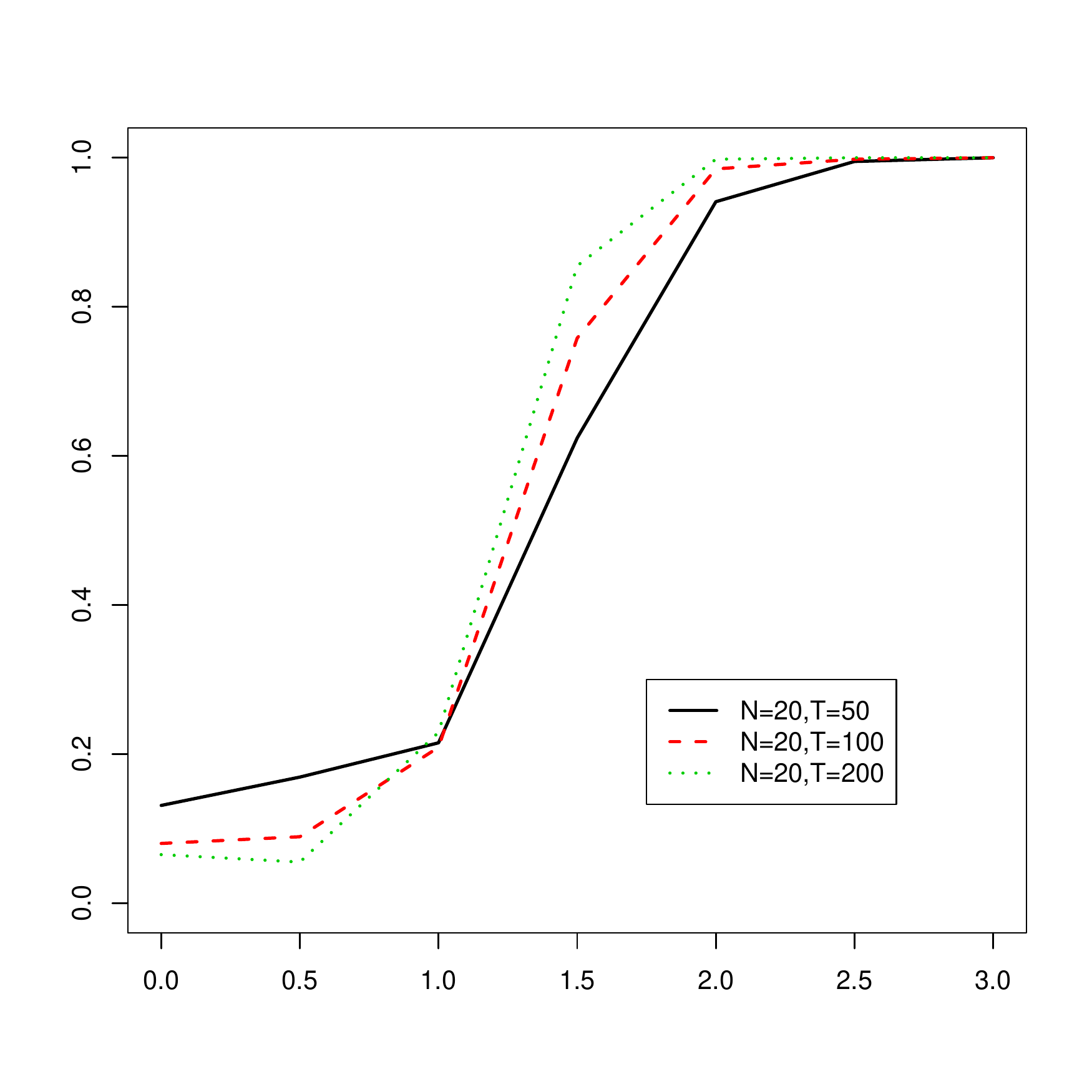}\\
      \caption{ Power curves generated from data following MB$(\delta)$ for fixed $N$ and varying $T$. The horizontal axis measures $\delta$, and the vertical axis measures the empirical power when the significance level is fixed at 5\%.}\label{Mpower-Nf}
\end{figure}

\begin{figure}[h!]
\centering
      \includegraphics[width=0.6\textwidth,height=0.6\textwidth]{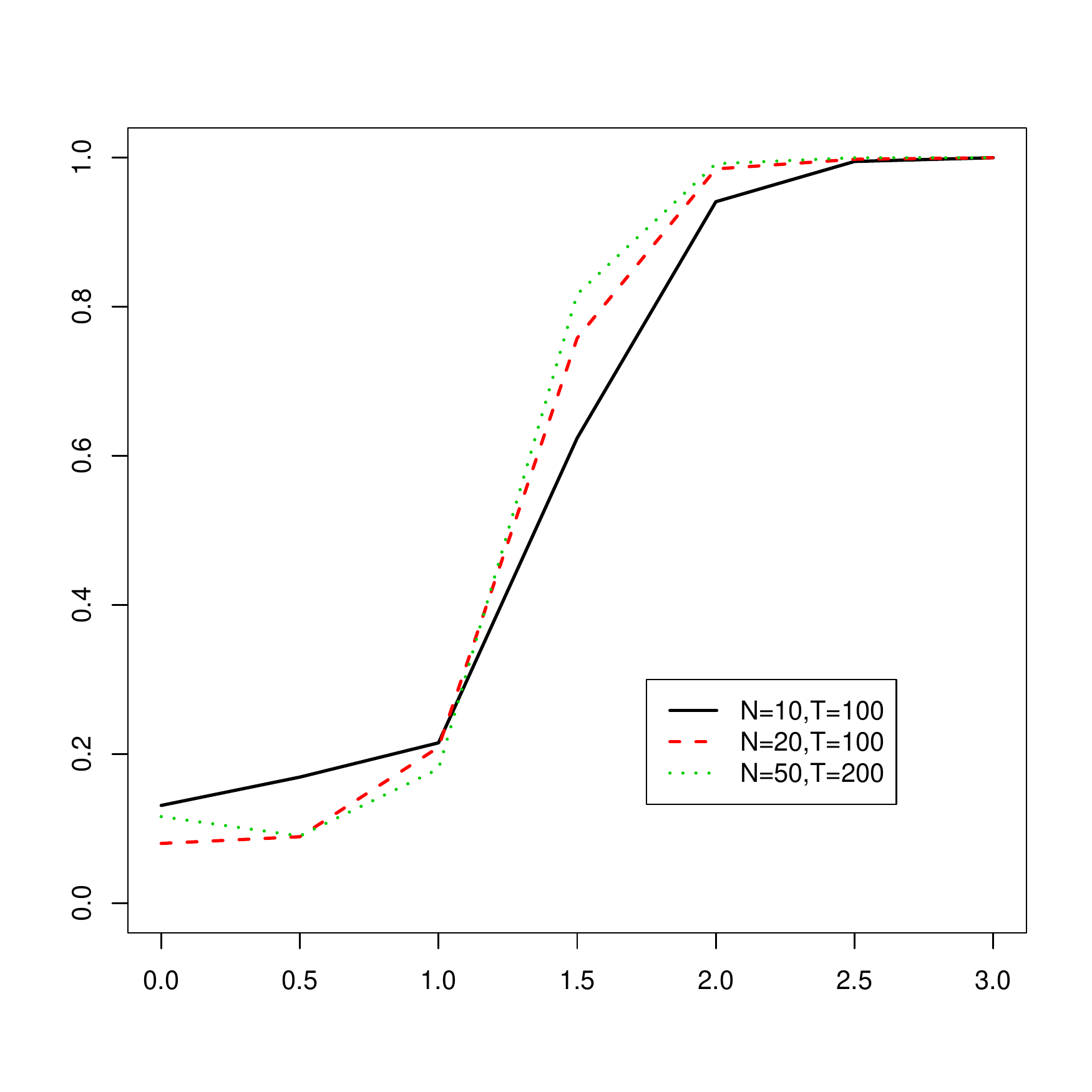}\\
      \caption{Power curves generated from data following MB$(\delta)$ for fixed $T$ and varying $N$. The horizontal axis measures $\delta$, and the vertical axis measures the empirical power when the significance level is fixed at 5\% .}\label{Mpower-Tf}
\end{figure}

\begin{figure}[h!]
\centering
      \includegraphics[width=0.6\textwidth,height=0.6\textwidth]{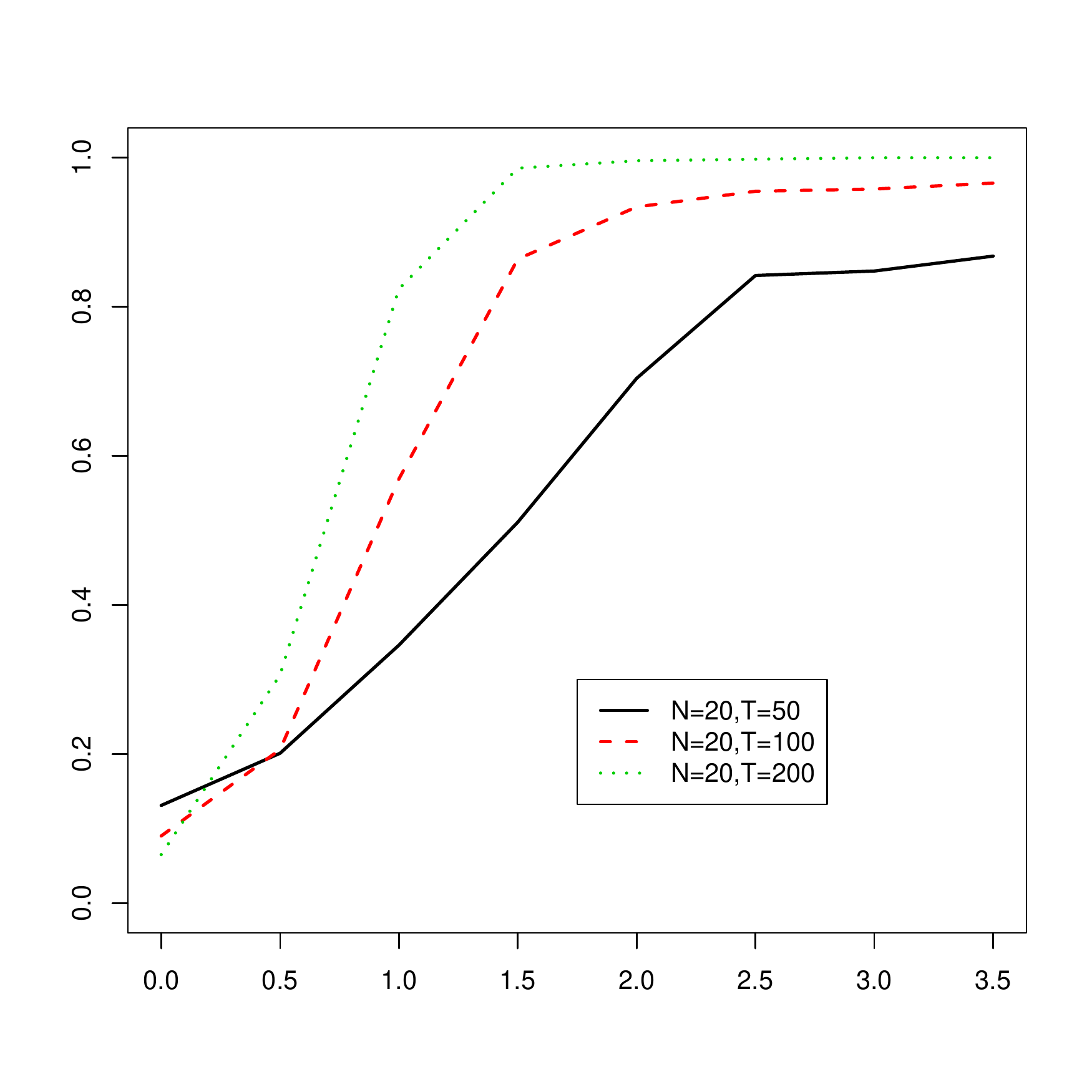}\\
      \caption{ Power curves generated from data following LB$(\Delta)$ for fixed $N$ and varying $T$. The horizontal axis measures $\Delta$, and the vertical axis measures the empirical power when the significance level is fixed at 5\%.}\label{Lpower-Nf}
\end{figure}

\begin{figure}[h!]
\centering
      \includegraphics[width=0.6\textwidth,height=0.6\textwidth]{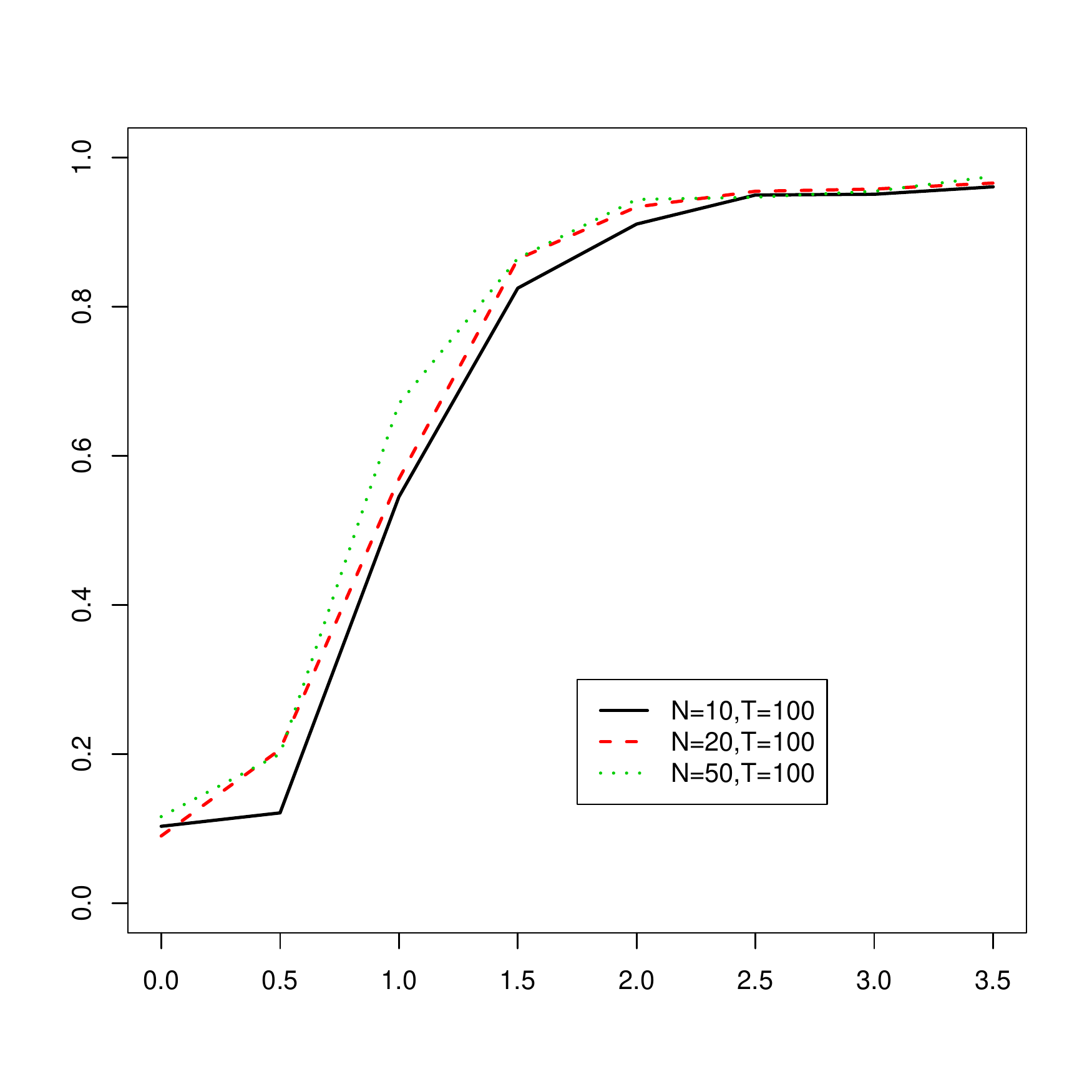}\\
      \caption{Power curves generated from data following LB$(\Delta)$ for fixed $T$ and varying $N$. The horizontal axis measures $\Delta$, and the vertical axis measures the empirical power when the significance level is fixed at 5\%.}\label{Lpower-Tf}
\end{figure}

{\it Mean Break:}
\begin{enumerate}

\item In the case of a mean break, for each value of $T$ and $N$ that we considered there is a substantial gain in power for $\delta$ exceeding 1.5. We note that data generated according to AR-1 have cross-sectional standard deviations of on average 1.6, and, when $\delta=2$, the average squared size of the change in the mean of each cross section is 1.33. Thus testing based on the largest eigenvalue seemed very sensitive to detect changes in the mean.

\item Due to the estimation of the variance under a mean break, the test exhibited monotonic power.

\item Increasing $T$ with fixed $N$ improved the empirical power, as expected, and the same was observed when $T$ was fixed and $N$ increased. The latter occurrence is likely attributable to the fact that as $N$ increases, changes in the mean occur in more cross sections, and the size is inflated in these cases due to the spiked effect.

\end{enumerate}
{\it Loading Break}
\begin{enumerate}
\item In the case of a break in the factor loadings, even smaller changes relative to the size of the standard deviation $(\Delta=1)$ of the idiosyncratic errors resulted in dramatic increases in power.

\item We noticed that for smaller values of $T$ $(T=50)$ the power seemed to level off for larger breaks in the common factors, and never reached more than $90\%$.

\item For larger $T$ $(T=100,200)$, the power approached 1 at a much faster rate for breaks in the factor loadings, and this occurrence seemed to be independent of the value of $N$.

\item Increasing $N$ resulted in reduced power in this case, although the effects of changing $N$ were not particularly pronounced.

\end{enumerate}

\section{Application to U.S. Yield Curve Data} \label{app}
Following Yamamoto and Tanaka (2015), we consider an application of our methodology to test for structural breaks in U.S. Treasury yield curve data considered in G\"urkaynak et al. (2007), which is available at {\tt http://www.federalreserve.gov/ \\ econresdata}, and which the authors graciously maintain. The data consists of yields for fixed interest securities with maturities between one and thirty years with one year increments ($N=30$). We studied a portion of this data set spanning from January 1st, 1990 to August 28th, 2015, that we further reduced from daily to monthly observations by considering only the data from the last day of each month. Figure \ref{ycs} illustrates the yield curves corresponding to 1, 5, 10, and 30 year maturities.

\begin{figure}[h!]
\centering
      \includegraphics[width=0.6\textwidth,height=0.6\textwidth]{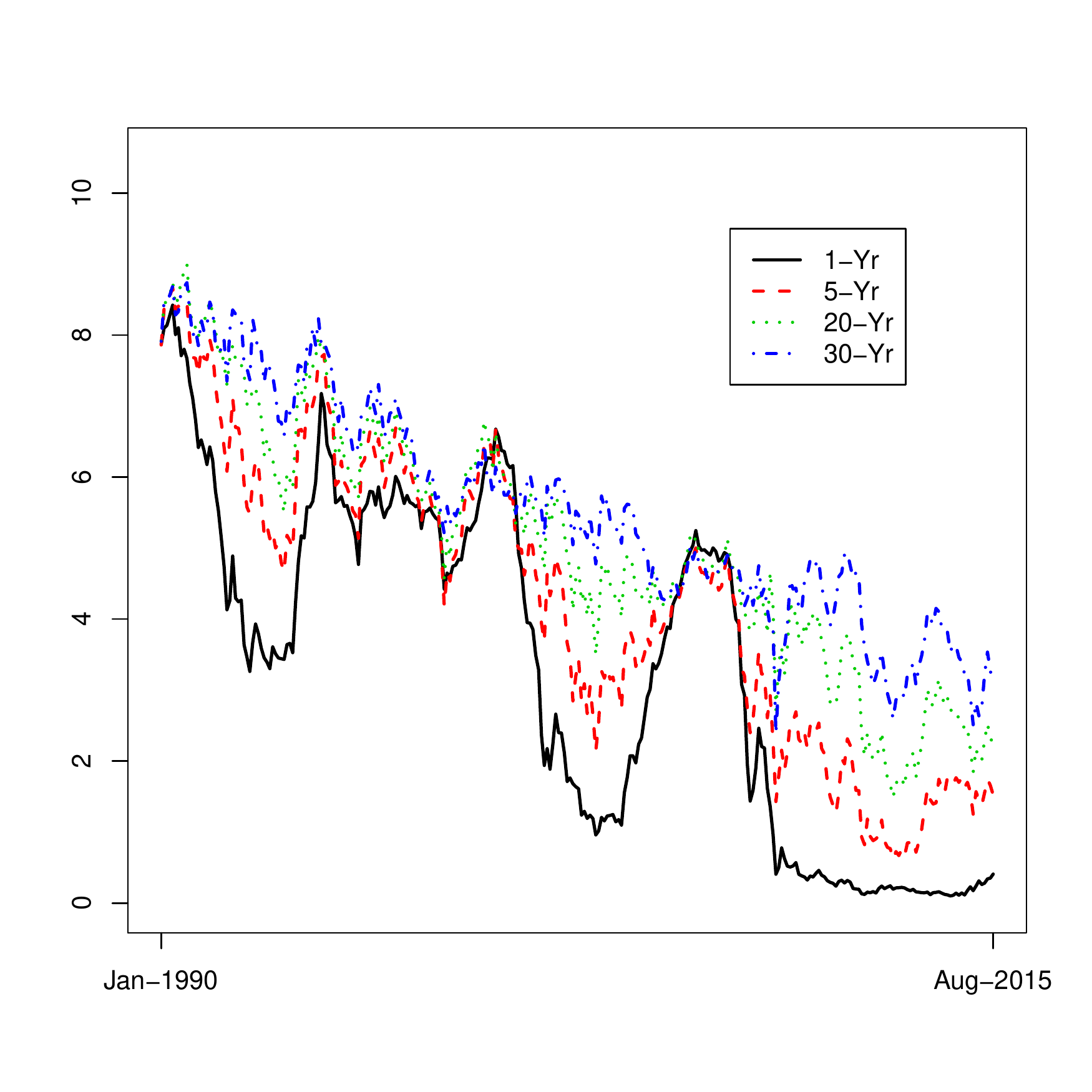}\\
      \caption{ Yield curves at a 1-month resolution between January, 1990 and August, 2015 correpsonding to 1 year, 5 year, 10 year, and 30 year maturities.}\label{ycs}
\end{figure}

In order to remove the effects of stochastic trends, and to allow for a comparison of our results to Yamamoto and Tanaka (2015), we first differenced each series. We applied the hypothesis test for stability of the largest eigenvalue based on $\sup_{0 \le t \le 1} |\tilde{B}_{T,1}(t)|$ with trimming parameter $\epsilon=.05$ to sequential blocks of the first differenced data of length 10 years, corresponding to 120 monthly observations in each sample ($T=120$). The first block contained data spanning from January, 1990 to December, 1999, and the last block contained data spanning from September, 2005 to August, 2015, which constituted a total of 172 tests. The P-value from each test is plotted against the end date of the corresponding 10 year block in Figure \ref{p-vals}.

\begin{figure}[h!]
\centering
      \includegraphics[width=0.6\textwidth,height=0.6\textwidth]{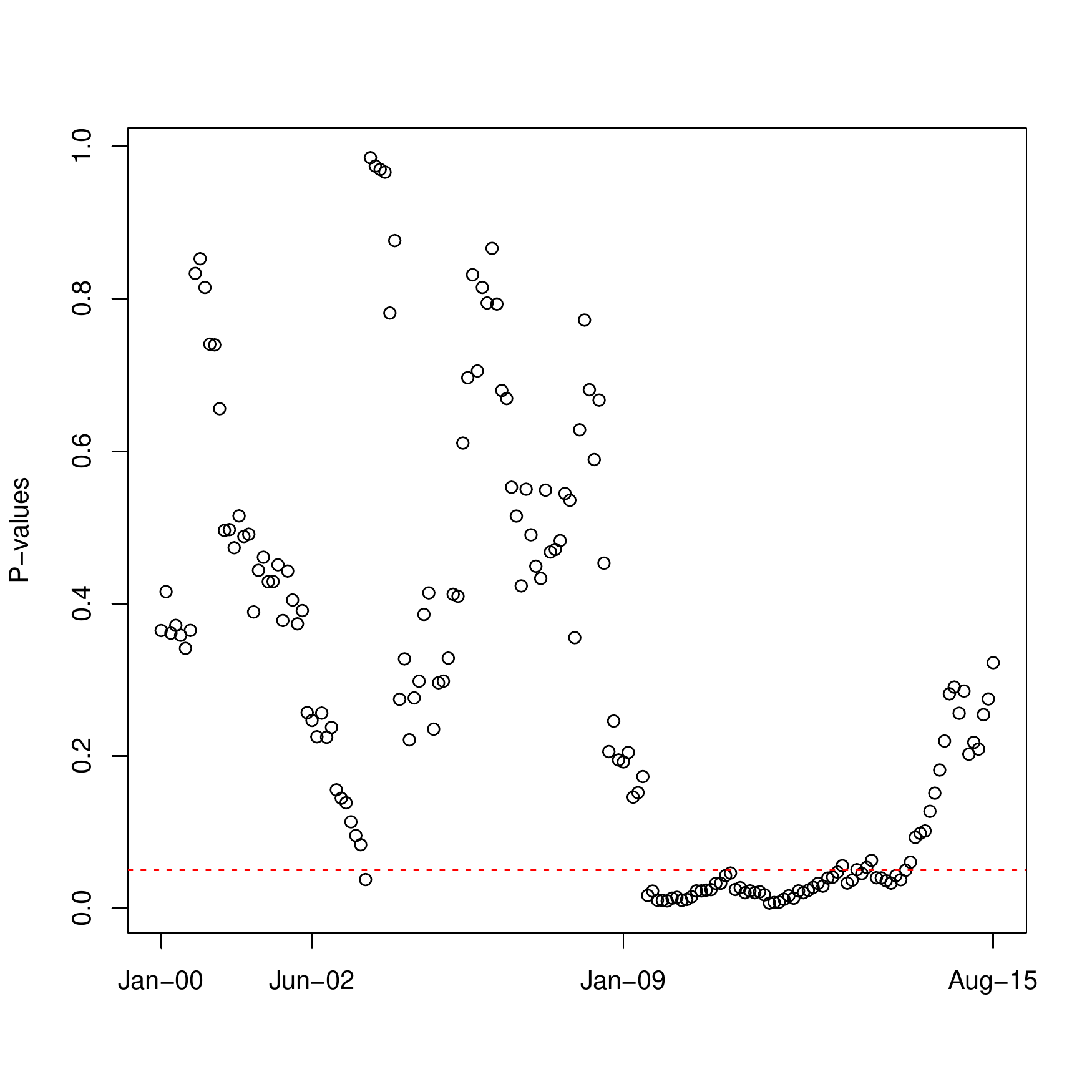}\\
      \caption{ P-values corresponding to 10 year blocks of the first differenced yield curve data. The vertical axis measures the magnitude of the P-value, and the horizontal axis indicates the concluding month of the 10 year block. P-values below the horizontal line are below .05.  }\label{p-vals}
\end{figure}

The most notable result of this analysis is the persistent instability of the largest eigenvalue evident in the samples that end in late 2008 to early 2009. This seems to correspond with the subprime crisis, which sparked what has been termed the ``Great Recession". The stability of the largest eigenvalue seems to return near the end of 2013. This may be indicative of the economic recovery, and provides a way of dating the end of the recession. The findings of structural breaks in the correlation structure of the yield curves during the 2007-2009 recession are consistent with those of Yamamoto and Tanaka (2015).

Also notable is the lack of persistent instability in relation to the 2001 economic recession. This illuminates a difference between the two recessions: The 2001 recession may be better modeled as a first order structural break, which is not as evident in the first differenced yield curve series, whilst the 2009 recession, which generated numerous policy changes and endured for a longer period, is manifested as a change in the largest eigenvalue.

\section{Results for smaller eigenvalues}\label{all-eig}

In this section, we provide analogous results to Theorems \ref{main} and \ref{main-2} for the smaller eigenvalues. Namely, we aim to establish the weak convergence of the $K$--dimensional process

$$
\A_{N,T}(u)=(A_{N,T,1}(u), A_{N,T,2}(u), \ldots ,A_{N,T,K}(u))^\T,
$$

where for $1\leq i \leq K$,

$$
A_{N,T,i}(u)=T^{1/2}u(\hat{\lambda}_i(u)-\lambda_i),\;\;1/T\leq u \leq 1\;\;\mbox{and}\;\;A_{N,T,i}(u)=0,\;\;0\leq u<1/T.
$$

Let
\beq\label{v-def-1}
V_1=\sum_{\ell=-\infty}^\infty\cov(\eta_0^2, \eta^2_\ell),
\eeq
\beq\label{v-def-2}
\V_2=\left\{\sum_{s=-\infty}^\infty\lim_{T\to \infty}\sum_{k=1}^N\fe_i(k)\fe_j(k)\cov(\eta_0, \eta_s)\cov(e_{k,0}, e_{k,s}), 1\leq i,j\leq K\right\},
\eeq
\begin{align}\label{v-def-3}
\V_3
&=\left\{\sum_{s=-\infty}^\infty\lim_{T\to \infty}\left[\sum_{k=1}^N\fe_i^2(k)\fe^2_j(k)\cov (e^2_{k,0}, e^2_{k,s}) \right.\right.\\
&\hspace{1.5cm}+2\left(\sum_{k=1}^N\fe_i(k)\fe_j(k)\cov (e_{k,0}, e_{k,s})\right)^2  \notag\\
&\hspace{1.5cm}\left. \left. -2\sum_{k=1}^N\fe_i^2(k)\fe_j^2(k)(\cov (e_{k,0}, e_{k,s}))^2
\right], 1\leq i,j\leq K
\right\}.\notag
\end{align}
We use the notation $\V_2=\{V_2(i,j), 1\leq i,j\leq K\}$ and $\V_3=\{V_3(i,j), 1\leq i,j\leq K\}$.
\medskip
\begin{rem}\label{rem-1} {\rm If, for example, we  assume that $r(s)=\cov (e_{k,0}, e_{k,s})$ for all $1\leq k\leq N$, then  $\V_2$ is a diagonal matrix with
$$
V_2(i,i)=\sum_{s=-\infty}^\infty \cov (\eta_0, \eta_s) r(s).
$$
The expression for $\V_3$ also simplifies since by the orthonormality of the $\fe_i$'s we have
$$
\sum_{k=1}^N\fe_i(k)\fe_j(k)\cov (e_{k,0}, e_{k,s})=r(s)I\{i=j\}.
$$
If we further assume that each of the $\{e_{k,s}, -\infty<s<\infty\}$ sequences are Gaussian, then $\mbox{cov} (e^2_{k,0}, e^2_{k,s})=2r^2(s)$, and $\V_3$ also reduces to a diagonal matrix with $V_3(i,i)=2\sum_{s=-\infty}^\infty r^2(s)$.
}
\end{rem}

\medskip

Let
\beq\label{a-def}
a_i=\lim_{T\to \infty}\fe_i^\T\bgamma,\;\;\;1\leq i\leq K
\eeq
and define $\G=\{G(i,j), 1\leq i,j\leq K\}$ with $G(i,j)=a_i^2a_j^2V_1+4a_ia_jV_2(i,j)+V_3(i,j)$.
Lemma \ref{lin-1} demonstrates that the limit in \eqref{a-def} is finite.

\begin{theorem}\label{main-all} If $H_0$ and the conditions of Theorem \ref{main} hold, and
\beq\label{ga-1}
\|\bgamma\|=O(1),\;\;\mbox{as}\;\;T\to\infty,
\eeq
then we have that
$\A_{N,T}(u)\;\;\mbox{converges weakly in}\;\;{\mathcal D}^K[0,1]\;\;\mbox{to}\;\;\W_{\G}(u),$
where $\W_{\G}(u)$ is a $K$--dimensional Wiener process, i.e.\ $\W_{\G}(u)$ is Gaussian with $E\W_{\G}(u)={\bf 0}$ and $E\W_{\G}(u)\W_{\G}^\T(u')=\min(u,u') \G$.
\end{theorem}

\medskip
\begin{rem}\label{rem-zero}{\rm If $\|\bgamma\|\to 0$, as $T\to \infty$, then $a_i=0$ according to Lemma \ref{lin-1}. In this case the weak limit of $\A_{N,T}(u)$ is the $K$--dimensional Wiener process $\W_{\V_3}(u)$, since  $\G=\V_3$.
}
\end{rem}

\medskip
 To state the next result we introduce the covariance matrix $\bH=\{H(i,j), 1\leq i ,j\leq K\}$: $H(1,1)=V_1, H(1,i)=H(i,1)=a_i^2V_1$ and $H(i,j)=a_i^2a_j^2V_1+4a_ia_j V_2(i,j)+ V_3(i,j), 2\leq i,j\leq K$.

\medskip
\begin{theorem}\label{main-all-2} If $H_0$ and the conditions of Theorem \ref{main} hold, and
\beq\label{ga-2}
\|\bgamma\|\to \infty,\;\;\mbox{as}\;\;T\to\infty,
\eeq
then we have that $\left\{\|\bgamma\|^{-2}A_{N,T;1}(u), A_{N,T;i}(u), 2\leq i\leq K \right\}$ converges weakly in
${\mathcal D}^K[0,1]$ to $\W_{\bH}(u)$,
where $\W_{\bH}(u)$ is a $K$--dimensional Wiener process, i.e.\ $\W_{\bH}(u)$ is Gaussian with $E\W_{\bH}(u)={\bf 0}$ and $E\W_{\bH}(u)\W_{\bH}^\T(u')=\min(u,u') \bH$.
\end{theorem}

\medskip

\begin{rem}\label{main-same}{\rm We note that $\sigma_1^2$ defined in Theorem \ref{main} coincide with $G(1,1)$ and $H(1,1)\|\gamma\|^2$ in the cases when $\|\bgamma\|=O(1)$ and $\|\bgamma\|\to \infty$ as $T \to \infty$, respectively, and so Theorems \ref{main-all} and \ref{main-all-2} imply Theorem \ref{main}.
}
\end{rem}

\begin{rem}\label{big-one}{\rm We show in Lemma \ref{lin-1}  that in case of \eqref{ga-2}, $\lambda_1$, the largest eigenvalue of $\C$ satisfies
$$
\left|\frac{\lambda_1}{\|\bgamma\|^2}-1\right|=O(1)
$$
Thus Theorem \ref{main-all-2} yields that $\hat{\lambda}_(u)/\|\bgamma\|^2\to 1$ in probability for all $u>0$.
}
\end{rem}
\medskip

\begin{rem}\label{nowe}{\rm Theorems \ref{main}, \ref{main-all} and \ref{main-all-2} provide the limits of the weighted differences $T^{1/2}u(\hat{\lambda}_i(u)$ $-\lambda_i)=T^{1/2}(\tilde{\lambda}_i-u\lambda_i), 1\leq i \leq K$. If the conditions of Theorem \ref{main} are satisfied but \eqref{n-t-1} is replaced \eqref{n-t-2} as in Theorem \ref{main-2}, then $T^{1/2}(\hat{\lambda}_i(u)-\lambda_i), 1\leq i \leq K$ converges weakly in ${\mathcal D}^K[c,1]$ to  $\W_{\G}(u)/u$ for any $0<c\leq 1$ where $\W_{\G}(u)$ is defined in Theorem \ref{main}.\\

}
\end{rem}

\medskip

\section{Technical Results }\label{proofs}

 \subsection{Proof of Theorems \ref{main}, \ref{main-2}, \ref{main-all} and \ref{main-all-2}}

Throughout these proofs we use the terms of the form $c_{i,j}$ to denote unimportant numerical constants. We can assume without loss of generality that $E\X_t={\bf 0}$, and so we  define
$$
\C_{N,T}(u)=\frac{1}{T}\sum_{t=1}^{\lf Tu\rf}\X_t\X_t^\T.
$$
\begin{lemma}\label{rem-m} If \eqref{model-null} and Assumptions \ref{inc}, \ref{b-g}, and \ref{error-as} hold, then we have, as $T\to\infty$,
$$
\sup_{0\leq u \leq 1}\left\|  \tC_{N,T}(u)-\C_{N,T}(u)  \right\|=O_P\left( \frac{N}{T}   \right).
$$
\end{lemma}
\begin{proof} It is easy to see that
$$
\tC_{N,T}(u)=\C_{N,T}(u)-\baX_T\left(\frac{1}{T}\sum_{t=1}^{\lf Tu\rf}\X_t\right)^\T-\left(\frac{1}{T}\sum_{t=1}^{\lf Tu\rf}\X_t\right)\baX_T^\T+
\baX_T\baX_T^\T,
$$
and therefore
\begin{align*}
\sup_{0\leq u \leq 1}\left\|  \tC_{N,T}(u)-\C_{N,T}(u)  \right\|&\leq 2\sup_{0\leq u \leq 1}\left\|  \baX_T\left(\frac{1}{T}\sum_{t=1}^{\lf Tu\rf}\X_t\right)^\T  \right\|+\left\|\baX_T\baX_T^\T\right\|\\
&\leq 3\sup_{0\leq u \leq 1}\left\|  \baX_T\left(\frac{1}{T}\sum_{t=1}^{\lf Tu\rf}\X_t\right)^\T  \right\|.
\end{align*}
Using assumption \eqref{model-null} we obtain that
\begin{align*}
&T^2\left\|\sum_{t=1}^{\lf Tu\rf}\X_t\baX_T^\T\right\|^2\\
&=\sum_{\ell=1}^N\sum_{p=1}^N \left(
\gamma_\ell\sum_{t=1}^{T}\eta_t+\sum_{t=1}^{T}e_{\ell,t}\right)^2\left(\gamma_p\sum_{t=1}^{\lf Tu\rf}\eta_t+\sum_{t=1}^{\lf Tu\rf}e_{p,t}\right)^2\\
&=\left(\sum_{\ell=1}^N\gamma_\ell^2\right)^2\left(\sum_{t=1}^{\lf Tu\rf}\eta_t\right)^2
\left(\sum_{t=1}^{ T}\eta_t\right)^2+2\sum_{\ell=1}^N\gamma_\ell^2\left(\sum_{t=1}^T\eta_t\right)^2\left(\sum_{v=1}^{\lf Tu\rf}\eta_v\right)
\left(\sum_{s=1}^{\lf Tu\rf}\sum_{p=1}^N\gamma_pe_{p,s}\right)\\
&\hspace{.5cm}+\sum_{\ell=1}^N\gamma_\ell^2\left(\sum_{t=1}^T\eta_t\right)^2\sum_{p=1}^N\left( \sum_{t=1}^{\lf Tu\rf}e_{p,t}  \right)^2
+\sum_{p=1}^N\gamma_p^2\left(\sum_{t=1}^{\lf Tu\rf}\eta_t\right)^2\sum_{\ell=1}^N\left(\sum_{s=1}^Te_{\ell,s}\right)^2
\\
&\hspace{.5cm}+2\sum_{\ell=1}^N\left(\sum_{s=1}^Te_{\ell,s}\right)^2\left(\sum_{s=1}^{\lf Tu\rf}\eta_s\right)\left(\sum_{s=1}^{\lf Tu\rf}\sum_{p=1}^N \gamma_p e_{p,s}\right)+\sum_{\ell=1}^N\left(\sum_{s=1}^Te_{\ell,s}\right)^2\sum_{p=1}^N\left( \sum_{s=1}^{\lf Tu\rf} e_{p,s}  \right)^2
\\
&\hspace{.5cm}+2\sum_{t=1}^T\eta_t\sum_{\ell=1}^N\gamma_\ell\left(\sum_{s=1}^Te_{\ell,s}\right)\sum_{p=1}^N\gamma_p^2\left(\sum_{v=1}^{\lf Tu\rf}\eta_v\right)^2+2\sum_{t=1}^T\eta_t\sum_{\ell=1}^N\gamma_\ell\left(\sum_{s=1}^Te_{\ell,s}\right)\sum_{p=1}^N\left( \sum_{s=1}^{\lf Tu\rf} e_{p,s}   \right)^2\\
&\hspace{.5cm}+4\sum_{t=1}^T\eta_t\sum_{\ell=1}^N\left(\sum_{s=1}^T\gamma_\ell e_{\ell,s}\right)\sum_{z=1}^{\lf Tu\rf}\eta_z\sum_{p=1}^N\left(\sum_{s=1}^{\lf Tu\rf}\gamma_p e_{p,s}\right)
\\
&=R_{T,1}(u)+R_{T,2}(u)+\ldots +R_{T,9}(u).
\end{align*}
First we prove that
\beq\label{l-1-1}
\sup_{0\leq u \leq 1}\left| \sum_{t=1}^{\lf Tu\rf}\eta_t \right|=O_P(T^{1/2}).
\eeq
It follows from Proposition 4 of Berkes et al.\ (2011) that under  conditions Assumption \ref{error-as}(a) and Assumption \ref{error-as}(a) we have for any $2<\kappa\leq 12 $ that
\beq\label{ber-mom}
E\left(\sum_{t=\lf Tv\rf}^{\lf Tu\rf}\eta_t\right)^\kappa\leq c_{1,1}(\lf Tu\rf-\lf Tv\rf)^{\kappa/2}\;\;\mbox{for all}\;\;0\leq v\leq u \leq 1,
\eeq
and therefore the maximal inequality of M\'oricz et al.\ (1982) implies \eqref{l-1-1}.
Next we show that
\beq\label{l-1-2}
\sup_{0\leq u \leq 1}\left|\sum_{s=1}^{\lf Tu\rf}\sum_{p=1}^N\gamma_pe_{p,s}\right|=O_P(1)T^{1/2}\|\bgamma\|.
\eeq
Following the arguments leading to \eqref{ber-mom} one can verify that for any $2<\kappa\leq 12$
\beq\label{ber-mom-2}
E\left|\sum_{s=\lf Tv\rf}^{\lf Tu\rf}e_{p,s}\right|^{{\kappa}}\leq c_{1,2}(\lf Tu\rf-\lf Tv\rf)^{\kappa/2}\;\;\;0\leq v\leq u \leq 1,
\eeq
with some constant $c_{1,2}$ for all $1\leq p \leq N$.
Hence for  any $0\leq v<u\leq 1$ we have via Rosenthal's inequality (cf.\ Petrov (1995), p.\ 59) and \eqref{ber-mom-2} that
\begin{align*}
E\left|\sum_{s=\lf Tv\rf}^{\lf Tu\rf}\sum_{p=1}^N\gamma_pe_{p,s}\right|^\kappa
&=E\left|\sum_{p=1}^N\sum_{s=\lf Tv\rf}^{\lf Tu\rf}\gamma_pe_{p,s}\right|^\kappa\\
&\leq c_{1,3}\left\{\sum_{p=1}^N |\gamma_p|^\kappa E\left|\sum_{s=\lf Tv\rf}^{\lf Tu\rf}e_{p,s}\right|^\kappa+\left(\sum_{p=1}^N\gamma_p^2 E\left(\sum_{s=\lf Tv\rf}^{\lf Tu\rf}e_{p,s}\right)^2\right)^{\kappa/2}
\right\}\\
&\leq c_{1,4} (\lf Tu\rf -\lf Tv\rf)^{\kappa/2}\left\{\sum_{p=1}^N |\gamma_p|^\kappa +\left(\sum_{p=1}^N \gamma_p^2 \right)^{\kappa/2}
\right\}.
\end{align*}
Using again the maximal inequality of M\'oricz et al.\ (1982) we conclude
\begin{align*}
E\sup_{0\leq u \leq 1}\left|\sum_{s=1}^{\lf Tu\rf}\sum_{p=1}^N\gamma_pe_{s,p}\right|^\kappa &\leq c_{1,5}T^{\kappa/2}\left\{\sum_{p=1}^N |\gamma_p|^\kappa +\left(\sum_{p=1}^N \gamma_p^2 \right)^{\kappa/2}
\right\}\\
&\leq c_{1,6} T^{\kappa/2}\left\|\bgamma\right\|^{\kappa},
\end{align*}
by Assumption \ref{b-g}. This completes the proof of  \eqref{l-1-2}.\\
Similarly to \eqref{l-1-2} we show that
\beq\label{l-1-3}
\sup_{0\leq s \leq 1}\sum_{\ell=1}^N\left(\sum_{s=1}^{\lf Tu\rf}e_{\ell,s}\right)^2=O_P(NT).
\eeq
First we note
$$
E\sup_{0\leq u \leq 1}\sum_{\ell=1}^N\left(\sum_{s=1}^{\lf Tu\rf}e_{\ell,s}\right)^2\leq \sum_{\ell=1}^NE\sup_{0\leq u \leq 1}\left(\sum_{s=1}^{\lf Tu\rf}e_{\ell,s}\right)^2
$$
and by Jensen's inequality we have
$$
E\sup_{0\leq u \leq 1}\left(\sum_{s=1}^{\lf Tu\rf}e_{\ell,s}\right)^2\leq \left(E\sup_{0\leq u \leq 1}\left|\sum_{s=1}^{\lf Tu\rf}e_{\ell,s}\right|^\kappa\right)^{2/\kappa}.
$$
Using again Proposition 4 of Berkes et al.\ (2011)  we get  for all $0\leq v \leq u \leq 1$ that
$$
E\left|\sum_{s=\lf Tv\rf}^{\lf Tu\rf}e_{\ell,s}\right|^\kappa\leq c_{1,7}(\lf Tu\rf-\lf Tv\rf)^{\kappa/2}
$$
and therefore the maximal inequality of M\'oricz et al.\ (1982)  yields
$$
\left(E\sup_{0\leq u \leq 1}\left|\sum_{s=1}^{\lf Tu\rf}e_{\ell,s}\right|^\kappa\right)^{2/\kappa}\leq c_{1,8}T^{1/2}.
$$
This completes the proof of \eqref{l-1-2}.\\
The upper bounds in \eqref{l-1-1}--\eqref{l-1-3} imply
$$
\sup_{0\leq u \leq 1}|R_{T,i}(u)|=O_P((\|\bgamma\|^4+\|\bgamma\|^3)T^2),\quad\mbox{if}\;\;i=1, 2, 7,
$$
$$
\sup_{0\leq u \leq 1}|R_{T,i}(u)|=O_P((\|\bgamma\|^2+\|\bgamma\|)NT^2),\quad\mbox{if}\;\;i=3, 4, 5, 8,  9.
$$
and
$$
\sup_{0\leq u \leq 1}|R_{T,6}(u)|=O_P(N^2T^2).
$$
Assumption \ref{b-g} implies that
$\|\bgamma\|\leq c_{1,9} N,$
  the proof of Lemma \ref{rem-m} is complete.
\end{proof}

\medskip
Let $\bar{\lambda}_1(u)\geq \bar{\lambda}_2(u)\geq \ldots \geq \bar{\lambda}_K(u)$ denote the $K$ largest eigenvalues of $\C_{N,T}(u)$.

\medskip
\begin{lemma}\label{eig-rem} If \eqref{model-null} and Assumptions \ref{inc}, \ref{b-g}, and \ref{error-as} hold, then we have, as $T\to\infty$,
$$
\max_{1\leq i \leq K}\sup_{0\leq u \leq 1}| \tilde{\lambda}_i(u)-\bar{\lambda}_i(u)|=O_P\left(\frac{N}{T}\right).
$$
\end{lemma}
\begin{proof} It is well--known (cf.\ Dunford and Schwartz (1988)) that
$$
\max_{1\leq i \leq K}\sup_{0\leq u \leq 1}| \tilde{\lambda}_i(u)-\bar{\lambda}_i(u)|\leq c_{2,1}\sup_{0\leq u\leq 1}\| \tilde{\C}_T(u)-\C_T(u)\|
$$
with some absolute constant $c_{2,1}$ and therefore the result follows from Lemma \ref{rem-m}.
\end{proof}

\medskip
Let
\begin{align*}
Z_{N,T;i}(u)
=\sum_{\ell\neq i}^N\frac{1}{u(\lambda_i-\lambda_\ell)}\left(\fe_i^\T(\tilde{\C}_{N,T}(u)-u\C)\fe_\ell\right)^2,\;\;1\leq i \leq K.
\end{align*}
\medskip

\begin{lemma}\label{hall} If \eqref{model-null}, Assumptions \ref{inc}, \ref{b-g}, and \ref{error-as} hold, then we have, as $T\to\infty$,
\begin{align*}
\sup_{0\leq u \leq 1}\left|\tilde{\lambda}_i(u)-\frac{\lf Tu\rf}{T}\lambda_i-\fe_i^\T(\tilde{\C}_{N,T}(u)-u\C)\fe_i-
Z_{N,T;i}(u)\right|
=O_P(NT^{-3/2}).
\end{align*}
\end{lemma}
\begin{proof} According to formula (5.17) of Hall and Hosseini--Nasab (2009) we have for all $1/T\leq u\leq 1$ that
\begin{align*}
&\left|\hat{\lambda}_i(u)-\lambda_i-\fe_i^\T(\hat{\C}_{N,T}(u)-\C)\fe_i-
\sum_{\ell\neq i}^N\frac{1}{\lambda_i-\lambda_\ell}\left(\fe_i^\T(\hat{\C}_{N,T}(u)-\C)\fe_\ell\right)^2
\right|\\
&\hspace{.5cm}\leq c_{3,1} \hat{\Delta}^{3}(u),
\end{align*}
where
$$
\hat{\Delta}(u)=\max_{1\leq \ell \leq N}\left(  \sum_{j=1}^N(\hat{C}_{N,T;j,\ell}(u)  -C_{j,\ell})^2 \right)^{1/2},
$$
and $\hat{C}_{N,T;j,\ell}(u)$ and $ C_{j,\ell}$ denote the $(k,\ell)^{\mbox{th}}$ element of $\hat{\C}_{N,T}(u)$ and $\C$, respectively. Hence
\begin{align*}
\sup_{0\leq u \leq 1}\left|\tilde{\lambda}_i(u)-\frac{\lf Tu\rf}{T}\lambda_i-\fe_i^\T(\tilde{\C}_{N,T}(u)-u\C)\fe_i-Z_{N,T;i}(u)
\right|
\leq c_{3,1}\sup_{0\leq u \leq 1}\Delta^{3}(u),
\end{align*}
where
$$
\Delta(u)=\max_{1\leq \ell \leq N} R_{N,T;\ell}(u)
$$
with
$$
R_{N,T;\ell}(u)=\left(  \sum_{j=1}^N(\tilde{C}_{N,T;j,\ell}(u)  -\frac{\lf Tu\rf}{T}C_{j,\ell})^2 \right)^{1/2},
$$
where $\tilde{C}_{N,T;j,\ell}(u)$ denotes the $(j,\ell)^{\mbox{th}}$ element of the matrix $\tilde{\C}_{N,T}(u)$.
By inequality (2.30) in Petrov (1995, p.\ 58) we conclude
$$
R_{N,T;\ell}^6(u)\leq N^2\sum_{j=1}^N\left(\tilde{C}_{N,T;j,\ell}(u)  -\frac{\lf Tu\rf}{T}C_{j,\ell}\right)^6
$$
and hence
$$
E\sup_{0\leq u \leq 1}\left( R_{N,T;\ell}(u)  \right)^6\leq N^2\sum_{j=1}^NE\sup_{0\leq u \leq 1}\left(\tilde{C}_{N,T;j,\ell}(u)- \frac{\lf Tu\rf}{T}C_{j,\ell}\right)^6.
$$
Using the definitions of $\tilde{C}_{N,T;j,\ell}(u)$ and $ C_{j,\ell}$ we write
\begin{align*}
\biggl(\tilde{C}_{N,T;j,\ell} &(u)-\frac{\lf Tu \rf}{T}C_{j,\ell} \biggl)^6\\
&=T^{-6}\left(\sum_{s=1}^{\lf Tu \rf}\left[\gamma_\ell\gamma_j(\eta_s^2-1)+\gamma_\ell\eta_se_{j,s}+\gamma_j\eta_se_{\ell,s}+e_{\ell,s}e_{j,s}-Ee_{\ell,s}e_{j,s}
\right]\right)^6\\
&\leq 4^6T^{-6}\left[\gamma_\ell^6\gamma_j^6\left(\sum_{s=1}^{\lf Tu\rf }(\eta_s^2-1)\right)^6+\gamma_\ell^6\left(\sum_{s=1}^{\lf Tu\rf}\eta_se_{j,s}\right)^6
+\gamma_j^6\left(\sum_{s=1}^{\lf Tu\rf}\eta_se_{\ell,s}\right)^6\right.\\
&\hspace{2cm}\left. +\left( \sum_{s=1}^{\lf Tu\rf}(e_{\ell,s}e_{j,s}- Ee_{\ell,s}e_{j,s})  \right)^6\right].
\end{align*}
Utilizing Assumption \ref{error-as}(a), we obtain along the lines of \eqref{ber-mom} that $E(\sum_{s=1}^t(\eta_s^2-1))^6\leq c_{3,2} t^3$, so by the stationarity of $\eta_t^2,-\infty<t<\infty$ and the maximal inequality of M\'oricz et al.\ (1982) we obtain that
$$
E\sup_{0\leq u \leq 1}\left(\sum_{s=1}^{\lf Tu\rf }(\eta_s^2-1)\right)^6\leq c_{3,3}T^3
$$
Similarly, for all $1\leq j, \ell\leq N$
$$
E\sup_{0\leq u \leq 1}\left(\sum_{s=1}^{\lf Tu\rf}\eta_se_{\ell,s}\right)^6\leq c_{3,4}T^3
\;\;\;\mbox{and}\;\;\;
E\sup_{0\leq u \leq 1}\left( \sum_{s=1}^{\lf Tu\rf}(e_{\ell,s}e_{j,s}- Ee_{\ell,s}e_{j,s})  \right)^6\leq c_{3,5}T^3.
$$
Hence for all $1\leq \ell\leq N$ we have by Assumption \ref{b-g} that
\beq\label{max-1}
E\left(\sup_{0\leq u \leq 1}R_{N,T;\ell}(u)\right)^6\leq c_{3,6}T^{-3}N^3.
\eeq
Using \eqref{max-1} we conclude for all $x>0$
\begin{align*}
P\left\{  \sup_{0\leq u \leq 1}\max_{1\leq \ell \leq N}R_{N,T;\ell}(u)>xN^{2/3}T^{-1/2}
\right\}
&\leq \sum_{\ell=1}^NP\left\{\sup_{0\leq u \leq 1}R_{N,T;\ell}(u)>xN^{2/3}T^{-1/2}
\right\} \\
&\leq \sum_{\ell=1}^N \frac{T^3}{x^6N^6}E\left(\sup_{0\leq u \leq 1}R_{N,T;\ell}(u)\right)^6\\
&\leq\sum_{\ell=1}^N \frac{T^3}{x^6N^6}C_5T^{-3}N^3,
\end{align*}
which shows that
\beq\label{D-1}
\sup_{0\leq u \leq 1}\Delta^3(u)=O_P(N^2T^{-3/2})\;\;\;\mbox{and}\;\;\;\sup_{0\leq u \leq 1}u\hat{\Delta}^3(u)=O_P(N^2T^{-3/2}).
\eeq
\end{proof}

\medskip
Since $\fe_1$ is defined via \eqref{ei-1} up to a sign, we can assume without loss of generality that $\bgamma^\T\fe_1\geq 0$.
\medskip

\begin{lemma}\label{lin-1} If \eqref{model-null}, Assumptions \ref{inc}, \ref{b-g}, and \ref{error-as} hold
and $\|\bgamma\|\to \infty$ hold, then we have
\beq\label{lin-11}
\left\|\fe_1-\frac{\bgamma}{\|\bgamma\|}\right\|=O\left(\frac{1}{\|\bgamma\|}\right),
\eeq
\beq\label{lin-123}
\frac{\lambda_1}{\|\bgamma\|^2}\to 1,
\eeq

\beq\label{lin-12}
\max_{2\leq i \leq N}|\bgamma^\T\fe_i|\leq c_{4,1}\;\;\mbox{with some constant}\;\;c_{4,1},
\eeq
and
\beq\label{lin-124}
\max_{2\leq i \leq N}\lambda_i\leq c_{4,2}\;\;\mbox{with some constant}\;\;c_{4,2}.
\eeq
\end{lemma}

\begin{proof} By \eqref{model-null} we have
$$
\C=\bgamma\bgamma^T+\bLambda,
$$
where $\bLambda$ is the $N\times N$ diagonal matrix with $\sigma_1^2, \sigma_2^2, \ldots, \sigma_N^2$ in the diagonal. We can write
$$
\fe_1=\bar{\alpha}_1\frac{\bgamma}{\|\bgamma\|}+\bar{\beta}_1\br_1,\;\;\mbox{with some}\;\;\bar{\alpha}_1\geq 0,\;\;\mbox{where}\;\; \bar{\alpha}_1^2+\bar{\beta}_1^2=1, \bgamma^\T\br_1=0\;\;\mbox{and}\;\;\|\br_1\|=1.
$$
It follows from the definition of $\lambda_1$ and $\fe_1$ that
\begin{align*}
\lambda_1=\fe_1^\T\C\fe_1\geq \|\bgamma\|^2
\end{align*}
and
\begin{align}\label{lin-125}
\fe_1^\T\C\fe_1=\bar{\alpha}_1^2\|\bgamma\|^2+\fe_1^\T\bLambda\C\fe_1,\;\;\fe_1^\T\bLambda\C\fe_1\leq \sum_{\ell=1}^N\fe_i^2(\ell)\sigma_\ell^2\leq c_5
\end{align}
where $c_5$ is defined in  Assumption \ref{error-as}(b). Thus we conclude
$$
\|\bgamma\|^2\leq \bar{\alpha}_1^2\|\bgamma\|^2+c_5.
$$
By assumption $\bgamma^\T\fe_1\geq 0$ and therefore $0\leq \bar{\alpha}_1\leq 1$. Hence  $(1-\bar{\alpha}_1)^2\leq 1-\bar{\alpha}_1^2\leq c_5/\|\bgamma\|^2$ and $\bar{\beta}_1^2\leq c_5/\|\bgamma\|^2$.  Thus we get
\beq\label{lin-13}
\left\|\fe_1-\frac{\bgamma}{\|\bgamma\|}\right\|^2=(1-\bar{\alpha}_1)^2+\bar{\beta}_1^2\leq \frac{2c_5}{\|\bgamma\|^2},
\eeq
completing the proof of \eqref{lin-11}. Since $\bar{\alpha}_2^2\geq 1-c_5/\|\bgamma|^2$, \eqref{lin-123} follows from \eqref{lin-125}.
 For all $i\geq 2$ we have
$$
|\bgamma^\T\fe_i|=\|\bgamma\|\left|\left(\frac{\bgamma}{\|\bgamma\|}-\fe_1\right)^\T\fe_i\right|\leq \|\bgamma\|\left\|\frac{\bgamma}{\|\bgamma\|}-\fe_1\right\|\leq 2c_5
$$
by \eqref{lin-13} which gives \eqref{lin-12}. Since $\lambda_i=\fe_i^\T\C\fe_i=(\fe_i^\T\bgamma)^2 +\fe_i^\T\bLambda\fe_i$ and $\fe_i^\T\bLambda\fe_i= \sum_{\ell=1}^N\fe_i^2(\ell)\sigma^2_\ell\leq c_5$ by Assumption \ref{error-as}(b), the last claim of this lemma follows from \eqref{lin-12}.
\end{proof}

\begin{lemma}\label{hall-2} If \eqref{model-null}, and Assumptions \ref{inc}, \ref{b-g}, and \ref{error-as} hold, then we have
\beq\label{Z-1}
\max_{1\leq i \leq K}\sup_{0\leq u\leq 1}|Z_{N,T;i}(u)|=O_P\left(\frac{N(\log T)^{1/3}}{T}\right).
\eeq
\end{lemma}

\begin{proof}
It follows from \eqref{ei-1} that $\fe_i\C\fe_\ell=0$, if $i\neq \ell$. Hence we get
\begin{align*}
\fe_i^\T\left(\tilde{\C}_{N,T}(u)-\frac{\lf Tu\rf}{T}\C\right)\fe_\ell&=\frac{1}{T}\sum_{s=1}^{\lf Tu\rf}\fe_i^\T\X_s\X_s^\T\fe_\ell.
\end{align*}
First we assume that $\|\bgamma\|=O(1)$.  It follows from the definition of $Z_{N,T;i}$ that
\begin{align*}
|Z_{N,T;i}(u)|&=\left|\sum_{\ell\neq i}^N\frac{1}{u(\lambda_i-\lambda_\ell)}\left(\fe_i^\T(\tC_{N,T}(u)-u\C)\fe_\ell  \right)^2\right|\\
&\leq \frac{1}{c_5}\frac{1}{T}\sum_{\ell\neq i}^N\left(\frac{1}{(Tu)^{1/2}}\sum_{s=1}^{\lf Tu\rf}\fe_i^\T \X_s\X_s^\T\fe_\ell\right)^2,
\end{align*}
where $c_0$ is defined in Assumption \ref{inc}. Let $\rho>1$ and write with $c=\lf 1/\log \rho\rf +1$
\begin{align*}
\max_{1\leq v \leq T}v^{-1/2}\left|\sum_{s=1}^{v}\fe_i^\T\X_s\X_s^\T\fe_\ell\right|&\leq \max_{1\leq k\leq c\log T}\max_{\rho^{k-1}<v\leq \rho^k}v^{-1/2}\left|\sum_{s=1}^{v}\fe_i^\T\X_s\X_s^\T\fe_\ell\right|\\
&\leq \max_{1\leq k\leq c\log T}\rho^{-(k-1)/2}\max_{1\leq v\leq \rho^k}\left|\sum_{s=1}^{v}\fe_i^\T\X_s\X_s^\T\fe_\ell\right|.
\end{align*}
Thus we get for any $x>0$ via Markov's inequality that
\begin{align}\label{maxx}
P&\left\{\max_{1\leq v \leq T}v^{-1/2}\left|\sum_{s=1}^{v}\fe_i^\T\X_s\X_s^\T\fe_\ell\right|>x\right\}\\
&\hspace{1cm}\leq
\sum_{k=1}^{c\log T}P\left\{\max_{1\leq v\leq \rho^k}\left|\sum_{s=1}^{v}\fe_i^\T\X_s\X_s^\T\fe_\ell\right|>x\rho^{(k-1)/2}
\right\}\notag\\
&\hspace{1cm}\leq \sum_{k=1}^{c\log T}x^{-6}\rho^{-3(k-1)}E\left(   \max_{1\leq v\leq \rho^k}\left|\sum_{s=1}^{v}\fe_i^\T\X_s\X_s^\T\fe_\ell\right|    \right)^6.\notag
\end{align}
Using \eqref{model-null} we obtain with $\fe_i=(\fe_i(1), \fe_i(2), \ldots, \fe_i(N))^\T$ that
\begin{align*}
\fe_i^\T\X_s\X_s^\T\fe_\ell
=&\sum_{k=1}^N\gamma_{k}\fe_i(k)\sum_{n=1}^N\gamma_{n}\fe_\ell(n)(\eta^2_s-1)
+\sum_{k=1}^N\gamma_{k}\fe_i(k)\eta_s\sum_{n=1}^Ne_{n,s}\fe_{\ell}(n)\\
&+\sum_{n=1}^N\gamma_{n}\fe_\ell(n)\eta_s\sum_{k=1}^Ne_{k,s}\fe_{i}(k)+\sum_{n=1}^N\sum_{k=1}^N(e_{k,s}e_{n,s}-Ee_{k,s}e_{n,s})\fe_{i}(k)\fe_\ell(n),
\end{align*}
since for $i\neq \ell$ we have $E\fe_i^\T\X_s\X_s^\T\fe_\ell=\fe_i^\T\C\fe_\ell=0$. Clearly, on account of $\|\fe_i\|=1$, the Cauchy--Schwarz inequality implies
\begin{align*}
\left|\sum_{k=1}^N\gamma_{k}\fe_i(k)\right|\leq \|\bgamma\|.
\end{align*}
Following the proofs of \eqref{ber-mom}, we get that from Assumption \ref{error-as}(a)  that
\begin{align}\label{fo-1}
E\left|\sum_{k=1}^N\gamma_{k}\fe_i(k)\sum_{m=1}^N\gamma_{m}\fe_m(\ell)\sum_{s=1}^v(\eta^2_s-1)\right|^6
\leq c_{5,1}v^{3}\|\bgamma\|^{12}.
\end{align}
Let
$$
\tau_{s}=\eta_s\sum_{n=1}^Ne_{n,s}\fe_\ell(n)\;\;\;\mbox{and}\;\;\;\tau_{s}^{(m)}=\eta_s^{(m)}\sum_{n=1}^Ne_{n,s}^{(m)}\fe_\ell(n),
$$
where $\eta_s^{(m)}$ and $e_{n,s}^{(m)}$ are defined in Assumption \ref{error-as}(a) and Assumption \ref{error-as}(b), respectively. By independence we have
\begin{align*}
E&\left|\tau_{0}-\tau_0^{(m)}\right|^{6}\\
&\leq 2^{6} E\left|\eta_0-\eta_0^{(m)}\right|^{6} E\left|\sum_{n=1}^Ne_{n,0}\fe_\ell(n)\right|^{6}
+2^{6} E\left|\eta_0^{(m)}\right|^{6} E\left|\sum_{n=1}^N(e_{n,0}-e_{n,0}^{(m)})\fe_\ell(n)\right|^{6}.
\end{align*}
By the independence of the variables $e_{n,0}, 1\leq n \leq N$  and the Rosenthal inequality (cf.\ Petrov (1995)) we conclude
\begin{align*}
E\left|\sum_{n=1}^Ne_{n,0}\fe_\ell(n)\right|^6 &\leq c_{5,2}\left\{\sum_{n=1}^NE|e_{n,0}|^6 |\fe_\ell(n)|^6+
\left(\sum_{n=1}^NE e_{n,0}^2\fe_\ell^2(n)  \right)^{3}
\right\}\\
&\leq c_{5,3}\sup_{1\leq n<\infty}Ee^6_{n,0}\\
&\leq c_{5,4},
\end{align*}
where $c_{5,4}$ is a constant, on account of Assumption \ref{error-as}(b) and $\|\fe_\ell\|=1$. Due to the independence of $e_{n,0}-e_{n,0}^{(m)}$ and $ e_{r,0}-e_{r,0}^{(m)}$, if $n\neq r$, we can apply again the Rosenthal inequality to get
\begin{align*}
E&\left|\sum_{n=1}^N(e_{n,0}-e_{n,0}^{(m)})\fe_\ell(n)\right|^6\\
&\leq c_{5,5}\left\{\sum_{n=1}^NE|e_{n,0}-e_{n,0}^{(m)}|^6 |\fe_\ell(n)|^6
+\left(\sum_{n=1}^NE(e_{n,0}-e_{n,0}^{(m)})^2\fe_\ell^2(n)\right)^{3}
\right\}\\
&\leq c_{5,6}m^{-6\alpha},
\end{align*}
resulting in
\begin{align}\label{hor-1}
E\left|\tau_{0}-\tau_0^{(m)}\right|^6\leq c_{5,7}m^{-6\alpha} .
\end{align}
Hence the moment inequality in Berkes et al.\  (2011) yields
\beq\label{b-1}
E\left|\sum_{s=1}^v\tau_s\right|^6\leq c_{5,8}v^{3}.
\eeq
Similarly to \eqref{b-1} we have
\beq\label{b-2}
E\left|\sum_{s=1}^v \eta_s\sum_{k=1}^Ne_{k,s}\fe_{i}(k) \right|^6\leq c_{5,9}v^{3}.
\eeq
Let
$$
\bar{\tau}_s=\sum_{n=1}^N\sum_{k=1}^N(e_{k,s}e_{n,s}-Ee_{k,s}e_{n,s})\fe_{i}(k)\fe_\ell(n)=\sum_{n=1}^Ne_{n,s}\fe_\ell(n)\sum_{k=1}^Ne_{k,s}\fe_{i}(k)
-\sum_{n=1}^NEe_{n,s}^2\fe_{i}(n)\fe_\ell(n)
$$
and
\begin{align*}
\bar{\tau}_s^{(m)}
&=\sum_{n=1}^N\sum_{k=1}^N(e_{k,s}^{(m)}e_{n,s}^{(m)}-Ee_{k,s}e_{n,s})\fe_{i}(k)\fe_\ell(n)\\
&=\sum_{n=1}^Ne_{n,s}^{(m)}\fe_\ell(n)\sum_{k=1}^Ne_{k,s}^{(m)}\fe_{i}(k)
-\sum_{n=1}^NEe_{n,s}^2\fe_{i}(n)\fe_\ell(n),
\end{align*}
where $e_{n,s}^{(m)}$ defined in Assumption \ref{error-as}(b). Clearly,
$$
\left|\sum_{n=1}^NEe_{n,s}^2\fe_{i}(n)\fe_\ell(n)\right|\leq \sup_{1\leq n<\infty}Ee_{n,0}^2,
$$
and
\begin{align*}
\bar{\tau}_s-\bar{\tau}_s^{(m)}=\left(\sum_{n=1}^N(e_{n,s}-e_{k,s}^{(m)})\fe_\ell(n)\right)\sum_{k=1}^Ne_{k,s}\fe_{i}(k)+\left(\sum_{k=1}^N(e_{k,s}-e_{k,s}^{(m)} )\fe_{i}(k)\right)\sum_{n=1}^Ne_{n,s}^{(m)}\fe_\ell(n).
\end{align*}
Thus we get by the Cauchy--Schwarz inequality that
\begin{align*}
E|\bar{\tau}_0-\bar{\tau}_0^{(m)}|^{6}\leq 2^{6} &\left\{\left(E\left| \sum_{n=1}^N(e_{n,0}-e_{k,0}^{(m)})\fe_\ell(n)\right|^{12}
E\left| \sum_{k=1}^Ne_{k,s}\fe_{i}(k)    \right|^{12}\right)^{1/2} \right.\\
& \hspace{1cm} +E\left. \left(\left|\sum_{k=1}^N(e_{k,0}-e_{k,0}^{(m)} )\fe_{i}(k)\right|^{12}E\left|\sum_{n=1}^Ne_{n,0}^{(m)}\fe_\ell(n)\right|^{12}\right)^{1/2}
\right\}.
\end{align*}
Using again  Rosenthal's and  Jensen's  inequalities, we obtain that
\begin{align*}
E\left| \sum_{n=1}^N(e_{n,0}-e_{k,0}^{(m)})\fe_\ell(n)\right|^{12}&\leq c_{5,10}\left\{\sum_{n=1}^NE|e_{n,0}-e_{k,0}^{(m)}|^{12}|\fe_\ell(n)|^{12}\right. \\
& \left. \hspace{1cm}+
\left( \sum_{n=1}^NE(e_{n,0}-e_{k,0}^{(m)})^2\fe_\ell^2(n)\right)^{6}
\right\}\\
&\leq c_{5, 11}m^{-12\alpha},
\end{align*}
and similarly
\begin{align*}
E\left| \sum_{k=1}^Ne_{k,0}\fe_{i}(k)    \right|^{12}&\leq c_{5, 12}\left\{\sum_{k=1}^NE|e_{k,0}|^{12}|\fe_{i}(k) |^{12}
+\left( \sum_{k=1}^NEe_{k,0}^2\fe_{i}^2(k)    \right)^{6}
\right\}\\
&\leq 2 c_{5,12}\sup_{1\leq k <\infty}E|e_{k,0}|^{12}.
\end{align*}
Thus we have
\beq\label{hor-2}
E|\bar{\tau}_0-\bar{\tau}_0^{(m)}|^{6}\leq c_{5, 13}m^{-6\alpha},
\eeq
and therefore Proposition 4 of Berkes et al.\ (2011) implies
\beq\label{b-3}
E\left|\sum_{s=1}^v\bar{\tau}_s\right|^6\leq c_{5, 14}v^3.
\eeq
Putting together \eqref{fo-1}--\eqref{b-3} we conclude
\beq\label{mor-1}
E\left|\sum_{s=1}^v\fe_i^\T\X_s\X_s^\T\fe_\ell\right|^6\leq c_{5, 15}v^3(1+\|\bgamma\|^6+\|\bgamma\|^{12}).
\eeq
Since $\fe_i^\T\X_s\X_s^\T\fe_\ell, -\infty<s<\infty$ is a stationary sequence, \eqref{mor-1} and the maximal inequality of M\'oricz et al. (1982) imply
\begin{align}\label{ineq**}
E\max_{1\leq v\leq z}\left|\sum_{s=1}^v\fe_i^\T\X_s\X_s^\T\fe_\ell\right|^6\leq c_{5, 16}z^3(1+\|\bgamma\|^6+\|\bgamma\|^{12}).
\end{align}
Now we use \eqref{maxx} with $x=u(\log T)^{1/6}$ resulting in
\begin{align*}
P\left\{\max_{1\leq v \leq T}v^{-1/2}\left|\sum_{s=1}^{v}\fe_i^\T\X_s\X_s^\T\fe_\ell\right|>u(\log T)^{1/6}\right\}\leq c_{5, 17}u^{-6},
\end{align*}
implying
$$
E\left(\max_{1\leq v \leq T}v^{-1/2}\sum_{s=1}^{v}\fe_i^\T\X_s\X_s^\T\fe_\ell\right)^2\leq c_{5, 18}(\log T)^{1/3}.
$$
This completes the proof of  \eqref{Z-1}.\\

Next we assume that $\|\bgamma\|\to \infty$. It is easy to see that for for  $2\leq i\leq K$
\begin{align*}
|Z_{N,T;i}(u)|\leq \frac{1}{T}&\left|\sum_{\ell\neq i, \ell\neq 1}^N\frac{1}{\lambda_i-\lambda_\ell}\left(\frac{1}{(Tu)^{1/2}}\sum_{s=1}^{\lf Tu\rf} \fe_i^\T\X_s\X_s^\T\fe_\ell^\T\right)^2\right|\\
&+\frac{1}{T}\frac{1}{\lambda_1-\lambda_2}\left(\frac{1}{(Tu)^{1/2}}\sum_{s=2}^{\lf Tu\rf} \fe_i^\T\X_s\X_s^\T\fe_1\right)^2.
\end{align*}
If  $2\leq i \leq K$, then the proof of \eqref{mor-1} shows that
$$
\sum_{\ell\neq i, \ell\neq 1}^N\left(\frac{1}{(Tu)^{1/2}}\sum_{s=1}^{\lf Tu\rf} \fe_i^\T\X_s\X_s^\T\fe_\ell\right)^2=O_P(N(\log T)^{1/3}),
$$
and therefore by Assumption \ref{inc} for any $2\leq i \leq K$ we have
\begin{align*}
\left|\sum_{\ell\neq i, \ell\neq 1}\frac{1}{\lambda_i-\lambda_\ell}\left(\frac{1}{(Tu)^{1/2}}\sum_{s=1}^{\lf Tu\rf} \fe_i^\T\X_s\X_s^\T\fe_\ell^\T\right)^2\right|=O_P(N(\log T)^{1/3}).
\end{align*}
By \eqref{b-3} we have along the lines of the proof of \eqref{maxx}
\begin{align}\label{horem-1}
E\sup_{1\leq v \leq T}\frac{1}{v}&\left(\sum_{s=1}^v\fe_i^\T\X_s\X_s^\T\fe_\ell -\bgamma^\T\fe_i\bgamma^\T\fe_\ell\sum_{s=1}^v(\eta_s^2-1)-
\bgamma^\T\fe_i\sum_{s=1}^v\sum_{n=1}^Ne_{n,s}\fe_\ell(n)  \right.\\
&\hspace{2cm}\left.-\bgamma^\T\fe_\ell\sum_{s=1}^v\sum_{k=1}^Ne_{k,s}\fe_i(k)\right)^2\leq c_{5, 19}(\log T)^{1/3},\notag
\end{align}
where in the last step we used \eqref{lin-12}.
 Also, \eqref{b-1} and \eqref{b-2} imply via the maximal inequality in M\'oricz et al.\ (1982) that
\beq\label{horem-2}
E\sup_{1\leq v \leq T}\left(\frac{1}{v}\sum_{s=1}^v(\eta_s^2-1)\right)^2\leq c_{5,20}(\log T)^{1/3},
\eeq
and
\beq\label{horem-3}
E\sup_{1\leq v \leq T}\frac{1}{v}\left( \sum_{s=1}^v\sum_{k=1}^Ne_{k,s}\fe_i(k)\right)^2\leq c_{5, 21}(\log T)^{1/3}.
\eeq
Using now  \eqref{horem-2} and \eqref{horem-3} we conclude that
\begin{align*}
\frac{1}{\lambda_1-\lambda_2}\left(\frac{1}{(Tu)^{1/2}}\sum_{s=1}^{\lf Tu\rf} \fe_i^\T\X_s\X_s^\T\fe_1\right)^2=\frac{(\fe_1^\T\bgamma)^2}{\lambda_1-\lambda_2}O_P((\log T)^{1/3}).
\end{align*}
Since by Lemma \ref{lin-1} we have that $(\fe_1^\T\bgamma)^2/(\lambda_1-\lambda_2)=O(1)$, the proof of \eqref{Z-1} is complete when $2\leq i \leq K$. It is easy to see that by \eqref{horem-1}  and Lemma \ref{lin-1}
\begin{align*}
\sup_{0\leq u \leq 1}|Z_{N,T;1}(u)|&\leq \frac{1}{T}\frac{1}{\lambda_1-\lambda_2}\sup_{0\leq u \leq 1}\sum_{\ell=2}^N\left(\frac{1}{(Tu)^{1/2}}\sum_{s=1}^{\lf Tu\rf}\fe_1^\T\X_s\X_s^\T\fe_\ell\right)^2\\
&=\frac{1}{T}\frac{N}{\lambda_1-\lambda_2}\left(O_P((\log T)^{1/3} +(\fe_1^\T\gamma)^2 E\max_{1\leq v \leq T}\left(v^{-1/2}\sum_{s=1}^v(\eta_s^2-1)\right)^2\right.\\
&\left.\hspace{1cm}
+E\max_{2\leq i\leq N}\left(v^{-1/2}\sum_{s=1}^v\sum_{k=1}^Ne_{k,s}\fe_i(k)\right)^2 \right)\\
&=\frac{(\fe_1^\T\gamma)^2}{\lambda_1-\lambda_2}\frac{N(\log T)^{1/3}}{T}
\end{align*}
an account of \eqref{horem-2} and \eqref{horem-3}. According to Lemma \ref{lin-1} we have that $(\fe_1^\T\gamma)^2/(\lambda_1-\lambda_2)=O(1)$, completing the proof of Lemma \ref{hall-2}.

\end{proof}

\medskip
Using the definition of $\tilde{\C}_{N,T}(u)$ and \eqref{model-null} we get for any $1\leq i \leq K$,
\begin{align*}
T\fe_i^\T(\tilde{\C}_{N,T}(u)-(\lf Tu\rf/T)\C)\fe_i=(\fe_i^\T\bgamma)^2&\sum_{t=1}^{\lf Tu\rf}(\eta_t^2-1)+2\fe_i^T\bgamma\sum_{t=1}^{\lf Tu\rf}\eta_t\sum_{\ell=1}^N\fe_i(\ell)e_{\ell,t}\\
&+\sum_{t=1}^{\lf Tu\rf}\left(\sum_{\ell=1}^N\fe_i(\ell)e_{\ell,t}\right)^2-\lf Tu\rf\sum_{\ell=1}^N\fe_i^2(\ell)\sigma^2_{\ell}.
\end{align*}
Let
$$
D_{N,T}(u)=\frac{1}{T^{1/2}}\sum_{t=1}^{\lf Tu\rf}(\eta_t^2-1),\;\;\;F_{N,T;i}(u)=\frac{1}{T^{1/2}}\sum_{t=1}^{\lf Tu\rf}\eta_t\sum_{\ell=1}^N\fe_i(\ell)e_{\ell,t},\;1\leq i\leq K,
$$
and
$$
G_{N,T;i}(u)=\frac{1}{T^{1/2}}\left\{\sum_{t=1}^{\lf Tu\rf}\left(\sum_{\ell=1}^N\fe_i(\ell)e_{\ell,t}\right)^2-\lf Tu\rf\sum_{\ell=1}^N\fe_i^2(\ell)\sigma^2_{\ell}\right\},\;1\leq i\leq K.
$$
\begin{lemma}\label{clt-lem} If  \eqref{model-null} and Assumptions \ref{inc}, \ref{b-g}, and \ref{error-as} hold, then $\{D_{N,T}(u), F_{N,T;i}(u), G_{N,T;i}(u), 0\leq u \leq 1, 1\leq i \leq K\}$ converges in ${\mathcal D}^{2K+1}[0,1]$ to the Gaussian process $\bGamma(u)=(\Gamma_1(u), \Gamma_2(u),\ldots, \\
\Gamma_{2K+1}(u))^\T,0\leq u \leq 1$, $E\bGamma(u)={\bf 0}$, and
\begin{displaymath}
E\bGamma(u)\bGamma^\T(u')=\min(u,u')
\left(
\begin{array}{ll}
V_1\; &{\bf 0}^\T\;\;\;\;\;\;\;{\bf 0}^\T
\vspace{.3cm}\\
{\bf 0}\;&\V_2\;\;\;\;\;\bO
\vspace{.3cm}\\
{\bf 0}\;&\bO\;\;\;\;\;\;\;\;\V_3
\end{array}
\right)
\end{displaymath}
\end{lemma}
\begin{proof} First we define the $m$--dependent processes
 $$
D_{N,T}^{(m)}(u)=\frac{1}{T^{1/2}}\sum_{t=1}^{\lf Tu\rf}((\eta_t^{(m)})^2-1),\;\;\;F_{N,T;i}^{(m)}(u)=\frac{1}{T^{1/2}}\sum_{t=1}^{\lf Tu\rf}\eta_t\sum_{\ell=1}^N\fe_i(\ell)e_{\ell,t}^{(m)},\;1\leq i\leq K,
$$
and
$$
G_{N,T;i}^{(m)}(u)=\frac{1}{T^{1/2}}\left\{\sum_{t=1}^{\lf Tu\rf}\left(\sum_{\ell=1}^N\fe_i(\ell)e_{\ell,t}^{(m)}\right)^2-\lf Tu\rf\sum_{\ell=1}^N\fe_i^2(\ell)\sigma^2_{\ell}\right\},\;1\leq i\leq K,
$$
where $\eta_t^{(m)}$ and $e_{\ell,t}^{(m)}$ are defined in Assumption \ref{error-as}(a) and Assumption \ref{error-as}(b), respectively. We show that for any $x>0$
\begin{align}\label{c-1}
\lim_{m\to\infty}\limsup_{T\to \infty}P\left\{|D_{N,T}(u)-D_{N,T}^{(m)}(u)|>x\right\}=0,
\end{align}
\begin{align}\label{c-2}
\lim_{m\to\infty}\limsup_{T\to \infty}P\left\{| F_{N,T;i}(u)-   F_{N,T;i}^{(m)}(u)   |>x\right\}=0,
\end{align}
and
\begin{align}\label{c-3}
\lim_{m\to\infty}\limsup_{T\to \infty}P\left\{|  G_{N,T;i}(u)-   G_{N,T;i}^{(m)}(u)   |>x\right\}=0,
\end{align}
for all $0<u\leq 1$  and $1\leq i \leq K$. It follows from Assumption \ref{error-as}(a) and the Cauchy--Schwarz inequality that
\begin{align}\label{c-4}
E\left|\eta_0^2-(\eta^{(m)}_0)^2\right|^6&= E\left\{|\eta_0+\eta^{(m)}_0||\eta_0-\eta^{(m)}_0|\right\}^6\\
&\leq 2^4 (E\eta_0^{12})^{1/2}(E|\eta_0-\eta^{(m)}|^{12})^{1/2}\notag\\
&\leq c_{6,1}m^{-6\alpha}.\notag
\end{align}
By stationarity, we get that
\begin{align*}
\mbox{var}&\left(T^{-1/2}\sum_{s=1}^{\lf Tu\rf}(\eta_s^2-(\eta^{(m)}_s)^2)^2\right)\\
&\leq \frac{1}{T}\sum_{s=1}^TE(\eta_s^2-(\eta^{(m)}_s)^2)
+2\sum_{s=1}^T E(\eta_0^2-(\eta^{(m)}_0)^2)(\eta_s^2-(\eta^{(m)}_s)^2)\\
&\leq E(\eta_0^2-(\eta^{(m)}_0)^2)^2+2\sum_{s=1}^T |E(\eta_0^2-(\eta^{(m)}_0)^2)(\eta_s^2-(\eta^{(m)}_s)^2)|.
\end{align*}
Since $\eta_0^2-(\eta^{(m)}_0)^2$ is independent of $\eta^{(m)}_s$,  if $s>m$, we obtain that
\begin{align*}
\sum_{s=m+1}^T &|E(\eta_0^2-(\eta^{(m)}_0)^2)(\eta_s^2-(\eta^{(m)}_s)^2)|\\
&\leq \sum_{s=m+1}^T |E(\eta_0^2-1)\eta_s^2|+\sum_{s=m+1}^T |E((\eta^{(m)}_0)^2)-1)\eta_s^2|.
\end{align*}
The independence of $\eta_0$ and $\eta_s^{(s)}$, \eqref{c-4}, and H\"older's inequality  yield
\begin{align*}
\sum_{s=m+1}^T |E(\eta_0^2-1)\eta_s^2|&=\sum_{s=m+1}^T |E(\eta_0^2-1)(\eta_s^2-(\eta_s^{(s)})^2|\\
&\leq \sum_{s=m+1}^\infty
(E|\eta_0^2-1|^{6/5})^{5/6}(E(\eta_0^2-(\eta_0^{(s)})^2)^6)^{1/6}\\
&\leq c_{6,2}m^{-(\alpha-1)}
\end{align*}
with $c_{6,2}=(c_{6,1}/(\alpha-1)) (E|\eta_0^2-1|^{6/5})^{5/6} $. The same argument gives that
$$
\sum_{s=m+1}^T |E((\eta^{(m)}_0)^2)-1)\eta_s^2|\leq c_{6,2}m^{-(\alpha-1)}.
$$
On the other hand, applying again \eqref{c-4} and the Cauchy--Schwarz inequality we conclude
\begin{align*}
\sum_{s=1}^m |E(\eta_0^2-(\eta^{(m)}_0)^2)(\eta_s^2-(\eta^{(m)}_s)^2)|\leq \sum_{s=1}^m E(\eta_0^2-(\eta^{(m)}_0)^2)^2 
\leq c_{6,1}m^{-(\alpha-1)}.
\end{align*}
Chebyshev's inequality now implies \eqref{c-1}. The proofs of \eqref{c-2} and \eqref{c-3} go along the lines of \eqref{c-1}, we only need to replace \eqref{c-4} with  \eqref{hor-1} and \eqref{hor-2}, respectively.
Next we show that for each $m$, $\{D_{N,T}^{(m)}(u), F_{N,T;i}^{(m)}(u), G_{N,T;i}^{(m)}(u), 0\leq u \leq 1, 1\leq i \leq K\}$ converges in ${\mathcal D}^{2K+1}[0,1]$ to the Gaussian process $\bGamma^{(m)}(u)=(\Gamma_1^{(m)}(u), \Gamma_2^{(m)}(u),\\
\ldots, \Gamma_{2K+1}^{(m)}(u))^\T,0\leq u \leq 1$, with  $E\bGamma^{(m)}(u)={\bf 0}$, and
\begin{displaymath}
E\bGamma^{(m)}(u)(\bGamma^{(m)})^\T(u')=\min(u,u')
\left(
\begin{array}{ll}
V_1^{(m)}\; &{\bf 0}^\T\;\;\;\;\;\;\;{\bf 0}^\T
\vspace{.3cm}\\
{\bf 0}\;&\V_2^{(m)}\;\;\;\;\;\bO
\vspace{.3cm}\\
{\bf 0}\;&\bO\;\;\;\;\;\;\;\;\V_3^{(m)}
\end{array}
\right)
\end{displaymath}
with
\beq\label{v-def-1-m}
V_1^{(m)}=\sum_{\ell=-m}^m\cov((\eta_0^{(m)})^2, (\eta^{(m)}_\ell)^2),
\eeq
\beq\label{v-def-2-m}
\V_2^{(m)}=\left\{\sum_{s=-m}^m\lim_{N\to \infty}\sum_{k=1}^N\fe_i(k)\fe_j(k)\cov(\eta_0^{(m)}, \eta_s^{(m)})\cov(e_{k,0}^{(m)}, e_{k,s}^{(m)}), 1\leq i,j\leq K\right\},
\eeq
and
\begin{align}\label{v-def-3-m}
\V_3^{(m)}
&=\left\{\sum_{s=-m}^m\lim_{N\to \infty}\left(\sum_{k=1}^N\fe_i^2(k)\fe^2_j(k)\cov ((e_{k,0}^{(m)})^2, (e_{k,s}^{(m)})^2) \right.\right.\\
&\hspace{1.5cm}+2\left(\sum_{k=1}^N\fe_i(k)\fe_j(k)\cov (e_{k,0}^{(m)}, e_{k,s}^{(m)})\right)^2  \notag\\
&\hspace{1.5cm}\left. \left. -2\sum_{k=1}^N\fe_i^2(k)\fe_j^2(k)(\cov (e_{k,0}^{(m)}, e_{k,s}^{(m)}))^2
\right), 1\leq i,j\leq K
\right\}.\notag
\end{align}
Let $0\leq u_1<u_2<\ldots <u_M\leq 1$ and $\mu_{i,k,\ell}, 1\leq i \leq M, 1\leq k,\ell\leq K$. We can write
\begin{align*}
&\sum_{k=1}^M \mu_{k,1,1}(D_{N,T}^{(m)}(u_k) -D_{N,T}^{(m)}(u_{k-1}))+\sum_{k=1}^M \sum_{i=1}^K \mu_{k,2,i}(F_{N,T,i}^{(m)}(u_k)-F_{N,T,i}^{(m)}(u_{k-1}))\\
&\hspace{2cm}+\sum_{k=1}^M \sum_{i=1}^K \mu_{k,3,i}(G_{N,T,i}^{(m)}(u_k)-G_{N,T,i}^{(m)}(u_{k-1}))\\
&={\mathcal S}_1+\ldots +{\mathcal S}_M,
\end{align*}
where
$$
{\mathcal S}_k=\sum_{s=\lf Tu_{i-1}\rf +1}^{\lf Tu_i\rf}\xi_{N,T;s}(k),\;\;1\leq i\leq M.
$$
The variables  $\xi_{N,T;s}(k),\lf Tu_{k-1}\rf +1\leq s \leq  \lf Tu_k\rf, 1\leq k \leq M$ are $m$--dependent and therefore $T^{-1/2}{\mathcal S}_1$, $T^{-1/2}{\mathcal S}_2,\ldots$, $T^{-1/2}{\mathcal S}_M$ are asymptotically independent. Hence we need only show the asymptotic normality of $T^{-1/2}{\mathcal S}_k$ for all $1\leq k \leq M$. For every fixed $k$ the variables $\xi_{N,T;s}(k),\lf Tu_{k-1}\rf +1\leq s \leq  \lf Tu_k\rf$ form an $m$--dependent stationary sequence with zero mean,
\begin{align*}
\lim_{T\to\infty}&\mbox{var}\left(T^{-1/2}{\mathcal S}_k\right)\\
&=\mbox{var}\left( \mu_{k,1,1} \Gamma_{1}^{(m)}(u_k) -(\Gamma_1^{(m)}(u_{k-1})) +
\sum_{i=1}^K \mu_{k,2,i}(\Gamma_{i+1}^{(m)}(u_k)-\Gamma_{i+1}^{(m)}(u_{k-1}))\right. \\
&\hspace{3cm}\left. +\sum_{i=1}^K \mu_{k,3,i}(\Gamma_{i+K+1}^{(m)}(u_k)-\Gamma_{i+K+1}^{(m)}(u_{k-1}))
 \right)
\end{align*}
and $E|\xi_{N,T;s}(k)|^3\leq C_{1}$, where $C_{1,1}$ does not depend on $N$ nor on $T$. Due to the $m$--dependence, these properties imply the asymptotic normality of $T^{-1/2}{\mathcal S}_k$. Applying the Cram\'er--Wold device (cf.\ Billingsley (1968)), we get that the finite dimensional distributions of $\{D_{N,T}^{(m)}(u),F_{N,T;i}^{(m)}(u), G_{N,T;i}^{(m)}(u), 0\leq u \leq 1, 1\leq i \leq K\}$ converge to that of  $\bGamma^{(m)}(u)$.
 Since $\|\V^{(m)}-\V\|\to 0$ as $T\to \infty$, and $\bGamma(u)$ and $\bGamma^{(m)}(u)$ are Gaussian processes we conclude that that $\bGamma^{(m)}(u)$ converges in ${\mathcal D}^{2K+1}[0,1]$ to $\bGamma(u)$. On account  of \eqref{c-1}--\eqref{c-3} we obtain that the finite dimensional distributions of $\{D_{N,T}(u),F_{N,T;i}(u), G_{N,T;i}(u), 0\leq u \leq 1, 1\leq i \leq K\}$ converge to that of $\bGamma(u)$. It is shown in the proof of Lemma \ref{rem-m}
 that
 $$
 E\left|\sum_{t=1}^v(\eta_t^2-1)\right|^3\leq c_{6,3}v^{3/2},\;\;  E \left|\sum_{t=1}^{v}\eta_t\sum_{\ell=1}^N\fe_i(\ell)e_{\ell,t}\right|^3\leq c_{6,4}v^{3/2}\;\;
 $$
 and
 $$
 E\left|\sum_{t=1}^{v}\left(\sum_{\ell=1}^N\fe_i(\ell)e_{\ell,t}\right)^2-v\sum_{\ell=1}^N\fe_i^2(\ell)\sigma^2_{\ell}\right|^3\leq c_{6,5}v^{3/2}.
 $$
 Due to the stationarity of $\eta, e_{i,t}, 1\leq i \leq N$, the tightness follows from Theorem 8.4 of Billingsley (1968).
\end{proof}

\medskip
{\it Proof of Theorem \ref{main-all}.} By Lemmas \ref{eig-rem}, \ref{hall} and \ref{hall-2}  we have that
$$
\sup_{0\leq u \leq 1}\left|T^{1/2}\left(\tilde{\lambda_i}(u)-\frac{\lf Tu\rf}{T}\lambda_i\right)-T^{1/2}\fe_i^\T\left(\tilde{\C}_{N,T}(u)-\frac{\lf Tu\rf}{T}\C\right)\fe_i\right|=o_P(1).
$$
Also,
\begin{align*}
\sup_{0\leq u \leq 1}&\left|T^{1/2}\fe_i^\T\left(\tilde{\C}_{N,T}(u)-\frac{\lf Tu\rf}{T}\C\right)\fe_i-G_{N,T;i}(u)\right|\\
&\leq
(\fe_i^\T\bgamma)^2\sup_{0\leq u \leq 1}|D_{N,T}(u)|+2|\fe_i^\T\bgamma|\sup_{0\leq u \leq 1}|F_{N,T;i}(u)|\\
&=O_P(1)((\fe_i^\T\bgamma)^2+|\fe_i^\T\bgamma|),
\end{align*}
since by Lemma \ref{clt-lem}
$$
\sup_{0\leq u \leq 1}|D_{N,T}(u)|=O_P(1)\;\;\;\mbox{and}\;\;\;\sup_{0\leq u \leq 1}|F_{N,T;i}(u)|=O_P(1).
$$
By the Cauchy--Schwarz inequality we have that $|\fe_i^\T\bgamma|\leq \|\bgamma\|$ and  therefore
$$
\sup_{0\leq u \leq 1}\left|T^{1/2}\fe_i^\T\left(\tilde{\C}_{N,T}(u)-\frac{\lf Tu\rf}{T}\C\right)\fe_i-G_{N,T;i}(u)\right|=o_P(1).
$$
The weak convergence of $G_{N,T;i}(u), 0\leq u \leq 1, 1\leq i \leq K$ is proven in Lemma \ref{clt-lem}, which completes the proof of Theorem \ref{main}.
\qed\\

{\it Proof of Theorem \ref{main-all-2}.}

Lemmas \ref{eig-rem} and \ref{hall} yield
\begin{align*}
\sup_{0\leq u \leq 1}\left|T^{1/2}\|\bgamma\|^{-2}\Bigl(\tilde{\lambda}_1(u)-
\frac{\lf Tu\rf}{T}\lambda_1\Bigl)-T^{1/2}\|\bgamma\|^{-2}\fe_1^\T\Bigl(\tilde{\C}_{N,T}(u)-\frac{\lf Tu\rf}{T}\C\Bigl)\fe_1\right|
=o_P(1).
\end{align*}
Thus Lemma \ref{clt-lem} yields
\begin{align*}
\sup_{0\leq u \leq 1}\left|T^{1/2}\|\bgamma\|^{-2}\left(\tilde{\lambda}_1(u)-\frac{\lf Tu\rf}{T}\lambda_1\right)-\frac{(\fe_1^\T\bgamma)^2}{\|\bgamma\|^2}D_{N,T}(u)
\right|=o_P(1).
\end{align*}
According to Lemma \ref{clt-lem} $\sup_{0\leq u \leq 1}|D_{N,T}(u)|=O_P(1)$ and  since  $(\fe_1^\T\bgamma)^2/\|\bgamma\|^2\to 1$ by Lemma \ref{lin-1}, we conclude
\begin{align}\label{fin-1}
\sup_{0\leq u \leq 1}\left|T^{1/2}\|\bgamma\|^{-2}\left(\tilde{\lambda}_1(u)-\frac{\lf Tu\rf}{T}\lambda_1\right)-D_{N,T}(u)
\right|=o_P(1).
\end{align}
  Lemmas \ref{lin-1} and  \ref{hall-2} imply
\begin{align}\label{fin-2}
\sup_{0\leq u \leq 1}&\left|T^{1/2}(\tilde{\lambda}_i(u)-u\lambda_i)-((\fe_i^\T\bgamma)^2D_{N,T}(u)+2\fe_i^\T\bgamma F_{N,T;i}(u)+G_{N,T;i}(u))
\right|\\
&=o_P(1).\notag
\end{align}
 Combining \eqref{fin-1} and \eqref{fin-2} with Lemma \ref{clt-lem}, we obtain that $\{T^{1/2}|\|\bgamma\|^{-2}(\tilde{\lambda}_1(u)-u\lambda_1), T^{1/2}(\tilde{\lambda}_i(u)-u\lambda_i), 2\leq i \leq K\}$ converges weakly in ${\mathcal D}^K[0,1]$ to $\bGamma^0(u)=(\Gamma^0_1(u), \Gamma^0_2(u),\\
 \ldots ,\Gamma^0_K(u))^\T$, where $\Gamma_1^0(u)=\Gamma_1(u)$ and
$\Gamma_i^0(u)=a_i^2\Gamma_1(u)+2a_i\Gamma_{i+1}(u)+\Gamma_{i+K+1}(u), 2\leq i \leq K$. The computation of the covariance function of $\bGamma^0(u)$ finishes the proof of Theorem \ref{main-all-2}.

\qed

\medskip

{\it Proof of Theorem \ref{main}.} Theorem \ref{main} is implied by Theorems \ref{main-all} and \ref{main-all-2} by Remark \ref{main-same}. 

\qed

{\it Proof of Theorem \ref{main-2} and Remark \ref{nowe}.} Theorem \ref{main-2} follows from Remark \ref{nowe}. Remark \ref{nowe} follows from Theorems \ref{main-all} and \ref{main-all-2} when the condition \eqref{n-t-1} is replaced with \eqref{n-t-2}. This requires replacing Lemma \ref{hall-2} with the result that for all $c>0$
$$
\max_{1\leq i \leq K}\sup_{c\leq u\leq 1}|Z_{N,T;i}(u)|=O_P\left(\frac{N}{T}\right),
$$
which follows from \eqref{ineq**} and Markov's inequality.

\qed

\subsection{Proof of Theorems \ref{estivar} and \ref{th-cons}}\label{cons-sec}

We prove a more general result concerning consistent estimates for norming sequences for each eigenvalue process. Let
$$
\hat{\xi}_{i,t}=(\hat{\fe}_i^\T(\X_t-\bar{\X}_T))^2,\;\;1\leq t\leq T\;\;\;1\leq i \leq K,
$$
and define
\begin{align*}
\hat{v}^2_{i,T}=\sum_{s=-N+1}^{N-1}J\left(\frac{s}{h}\right)\hat{r}_{i,s},
\end{align*}
where
\begin{displaymath}
\hat{r}_{i,s}=
\left\{
\begin{array}{ll}
\displaystyle \frac{1}{T-s}\sum_{t=1}^{T-s} (\hat{\xi}_{i,t}-\bar{\xi}_{i,T})(\hat{\xi}_{i,t+s}-\bar{\xi}_{i,T}),\;\;\mbox{if}\;\;s\geq 0
\vspace{.2cm}\\
\displaystyle\frac{1}{T-|s|}\sum_{t=-s}^{T} (\hat{\xi}_{i,t}-\bar{\xi}_{i,T})(\hat{\xi}_{i,t+s}-\bar{\xi}_{i,T}),\;\;\mbox{if}\;\;s< 0,
\end{array}
\right.
\end{displaymath}
where
$$
\bar{\xi}_{i,T}=\frac{1}{T}\sum_{t=1}^T\xi_{i,t}.
$$

We show that if $\|\bgamma\|=O(1)$ as $T\to \infty$, then
\beq\label{cos-1}
\frac{\hat{v}^2_{i,T}}{G(i,i)}\;\;\stackrel{P}{\to}\;1, \;\;\;\mbox{as}\;\;T\to \infty.
\eeq

Moreover, if $\|\bgamma\|\to \infty$ as $T \to \infty$, then

\beq\label{cos-2}
\frac{\hat{v}^2_{1,T}}{V_1\|\bgamma\|^4}\;\;\stackrel{P}{\to}\;1, \;\;\;\mbox{as}\;\;T\to \infty,
\eeq

and for $2\leq i \leq K$,

\beq\label{cos-3}
\frac{\hat{v}^2_{i,T}}{H(i,i)}\;\;\stackrel{P}{\to}\;1, \;\;\;\mbox{as}\;\;T\to \infty.
\eeq

We can assume without loss of generality that $E\bX_t={\bf 0}$. Elementary algebra gives that
\begin{align*}
(\bX_t-&\bar{\bX}_T)(  \bX_t-\bar{\bX}_T)^\T-\frac{1}{T}\sum_{u=1}^T(\bX_u-\bar{\bX}_T)(  \bX_u-\bar{\bX}_T)^\T\\
&=\bX_t\bX_t^\T-E\bX_0\bX_0^\T-\frac{1}{T}\sum_{u=1}^T\left(\bX_u\bX_s^\T-E\bX_0\bX_0^\T\right)
-\bX_t\bar{\bX}_T^\T-\bar{\bX}_T \bX_t^\T.
\end{align*}
It is easy to see that
\begin{align*}
&E\left|\hfe_i^\T\left[\frac{1}{T}\sum_{t=1}^T(\bX_t\bX_t^\T-E\bX_0\bX_0^\T)\right]\hfe_i\left[ \hfe_i^\T\frac{1}{T}\sum_{u=1}^T\left(\bX_u\bX_s^\T-E\bX_0\bX_0^\T\right)\hfe_i\right]\right|\\
&\hspace{.5cm}\leq E\left\|\frac{1}{T}\sum_{t=1}^T(\bX_t\bX_t^\T-E\bX_0\bX_0^\T)\right\|^2\\
&\hspace{.5cm}=
\frac{1}{T^2}\sum_{\ell=1}^N\sum_{\ell'=1}^NE\left(\sum_{u=1}^T(X_{\ell,u}X_{\ell',u}-EX_{\ell,u}X_{\ell',u})\right)^2\\
&=O\left(\frac{N^2}{T}\right)
\end{align*}
and  therefore by Markov's inequality we have
\begin{align}\label{ps-00}
&\lim_{T\to\infty}\left|\left[\frac{1}{T}\sum_{t=1}^T\hfe_i^\T(\bX_t\bX_t^\T-E\bX_0\bX_0^\T)\hfe_i\right]\left[ \hfe_i^\T\frac{1}{T}\sum_{u=1}^T\left(\bX_u\bX_s^\T-E\bX_0\bX_0^\T\right)\hfe_i\right]\right|\\
&\hspace{.5cm}=O_P\left(\frac{N^2}{T}\right).\notag
\end{align}
Using the same arguments as above, for every $c_{7,1}$ one can find  $c_{7,2}$ such that
\begin{align}\label{ps-0}
\lim_{T\to\infty}P\left\{\frac{1}{T}\sum_{t=1}^T\hfe_i^\T(\bX_t\bX_t^\T-E\bX_t\bX_t^\T)\hfe_i\hfe_i^\T\X_{t+s}\bar{\bX}_T\hfe_i^\T\geq c_{7,2}N^2/T\right\}\leq c_{7,1}
\end{align}
for every $c_{7,3}$ there is $c_{7,4}$ such that
\begin{align}\label{ps01}
\lim_{T\to\infty}P\left\{\left|\hfe_i^\T\frac{1}{T}\sum_{u=1}^T(\bX_u\bX_u^\T-E\bX_u\bX_u^\T)\hfe_i\frac{1}{T}\sum_{t=1}^T\hfe_i^\T\bX_{t+s}\bar{\X}_T\hfe_i
\right|\geq c_{7,4}N^2/T\right\}\leq c_{7,3}.
\end{align}
We note
\begin{align*}
\left|\frac{1}{T}\sum_{t=1}^{T-s}\hfe_i^\T\bX_t\bar{\bX}_T\hfe_i\hfe_i^\T\bX_{t+s}\bar{\bX}_T\hfe_i\right|\leq
\|\bar{\bX}_T\|^2\frac{1}{T}\sum_{t=1}^{T}\|\bX_t\|\|\bX_{t+s}\|.
\end{align*}
By \eqref{model-null} and assumption $\mu_i=0$ we get that from Assumption \ref{error-as}(a)--Assumption \ref{error-as}(b) and Assumption \ref{b-g}
\begin{align*}
E\|\bar{\bX}_T\|^2=\frac{1}{T^2}\sum_{u,v=1}^TE\bX_u^\T\bX_v&=\frac{1}{T^2}\sum_{u,v=1}^T\sum_{\ell=1}^NEX_{\ell,u}X_{\ell,v}\\
&=\frac{1}{T^2}\sum_{u,v=1}^T\sum_{\ell=1}^N(\gamma_\ell^2E\eta_u\eta_v+Ee_{\ell, u}e_{\ell,v})\\
&=O\left(\frac{N}{T}   \right)
\end{align*}
using the arguments in the proof of Lemma \ref{rem-m}. Due to stationarity we have
$$
E\|\bX_t\|\|\bX_{t+s}\|\leq (E\|\bX_t\|^2E\|\bX_{t+s}\|^2)^{1/2}=E\|\bX_0\|^2
$$
and
\begin{align*}
E\|\bX_0\|^2=\sum_{\ell=1}^N (\gamma_\ell^2 +Ee_{\ell,0}^2)=O(N)
\end{align*}
Hence for every $c_{7,5}$ there is $c_{7,6}$ such that
\begin{align}\label{ps-1}
\lim_{T\to\infty}P\left\{\left|\frac{1}{T}\sum_{t=1}^{T-s}\hfe_i^\T\bX_t\bar{\bX}_T\hfe_i\hfe_i^\T\bX_{t+s}\bar{\bX}_T\hfe_i\right|\geq c_{7,6}N^2/T\right\}\leq c_{7,5}.
\end{align}
Putting together \eqref{ps-00}--\eqref{ps01}  we conclude
\begin{align*}
\hat{v}_{i,T}^2=\tilde{v}_{i,T}^2+O_P\left(  \frac{hN^2}{T} \right),
\end{align*}
where
$$
\tilde{v}^2_{i,T}=\sum_{s=-N+1}^{N-1}J\left(\frac{s}{h}\right)\tilde{r}_{i,s},
$$
where
\begin{displaymath}
\tilde{r}_{i,s}=
\left\{
\begin{array}{ll}
\displaystyle \frac{1}{T-s}\sum_{t=1}^{T-s} \tilde{\xi}_{i,t}\tilde{\xi}_{i,t+s},\;\;\mbox{if}\;\;s\geq 0
\vspace{.2cm}\\
\displaystyle\frac{1}{T-|s|}\sum_{t=-s}^{T} \tilde{\xi}_{i,t}\tilde{\xi}_{i,t+s},\;\;\mbox{if}\;\;s< 0
\end{array}
\right.
\end{displaymath}
with $\tilde{\xi}_{i,t}=\hfe_i^\T(\bX_t\bX_t^\T-E(\bX_0\bX_0^\T))\hfe_i$.\\
It follows from Dunford and Schwartz (1988) and Assumption \ref{inc} that with some constant $c_{7,6}$
\begin{align}\label{dusch}
\max_{1\leq i \leq K}\|\hfe_i-\hat{c}_i\fe_i\|\leq c_{7,6}\|\hat{\C}_{N,T}(1)-\C\|,
\end{align}
where $\hat{c}_i, 1\leq i \leq K$ are random signs. We write
$$
\|\hat{\C}_{N,T}(1)-\C\|\leq \left\| \frac{1}{T}\sum_{t=1}^T(\X_t\X_t^\T-\C)\right\|+\left\| \bar{\X}_T\bar{\X}^\T  \right\|,
$$
and since we can assume without loss of generality that  $E\X_t={\bf 0}$ we get from the proof of Lemma \ref{rem-m}
$$
\left\| \bar{\X}_T\bar{\X}_T^\T  \right\|=O_P\left(\frac{N}{T}\right).
$$
Also,
\begin{align*}
E\left\| \sum_{t=1}^T(\X_t\X_t^\T-\C)\right\|^2&=E\sum_{\ell,\ell'=1}^N\left(\sum_{t=1}^T(X_{\ell,t}X_{\ell',t}-EX_{\ell,t}X_{\ell',t})\right)^2\\
&=\sum_{\ell,\ell'=1}^N\sum_{t,t'=1}^T(EX_{\ell,t}X_{\ell',t}X_{\ell,t'}X_{\ell',t'}-EX_{\ell,t}X_{\ell',t}EX_{\ell,t'}X_{\ell',t'}),
\end{align*}
\begin{displaymath}
EX_{\ell,t}X_{\ell',t}=\left\{
\begin{array}{ll}
\gamma_\ell\gamma_{\ell'},&\quad\mbox{if}\quad\ell\neq\ell'
\vspace{.2cm}\\
\gamma_{\ell}^2+Ee^2_{\ell,0},&\quad\mbox{if}\quad\ell=\ell'.
\end{array}
\right.
\end{displaymath}
and
\begin{displaymath}
EX_{\ell,t}X_{\ell',t}X_{\ell,t'}X_{\ell',t'}=\left\{
\begin{array}{ll}
\gamma_{\ell}^2\gamma_{\ell'}^2E\eta_t^2\eta_{t'}^2+\gamma_\ell^2E\eta_t\eta_{t'}Ee_{\ell',t}e_{\ell',t'}+
\gamma^2_{\ell'}E\eta_t\eta_{t'}Ee_{\ell,t}e_{\ell,t'}
\vspace{.2cm}\\
\;\;\;+Ee_{\ell,t}e_{\ell,t'} Ee_{\ell',t}e_{\ell',t'},\quad\mbox{if}\quad\ell\neq\ell'
\vspace{.2cm}\\
\gamma_{\ell}^4E\eta_t^2\eta_{t'}^2+2\gamma_\ell^2E\eta_0^2Ee^2_{\ell,0}+4\gamma_\ell^2E\eta_t\eta_{t'}Ee_{\ell, t}e_{\ell,t'}+Ee^2_{\ell,t}e^2_{\ell,t'},
\vspace{.2cm}\\
\quad\mbox{if}\quad\ell=\ell'.
\end{array}
\right.
\end{displaymath}
Thus we have
\begin{align*}
\sum_{\ell=1}^N&\sum_{t,t'=1}^T(EX^2_{\ell,t}X^2_{\ell,t'}-(EX^2_{\ell,0})^2)\\
&=\sum_{\ell=1}^N\gamma_\ell^4\sum_{t,t'=1}^T(E\eta_t^2\eta_{t'}^2-(E\eta_0^2)^2)
+  4\sum_{\ell=1}^N\gamma_\ell^2\sum_{t,t'=1}^TE\eta_t\eta_{t'}
Ee_{\ell,t}e_{\ell,t'}\\
&\hspace{1cm}+\sum_{\ell=1}^N\sum_{t,t'=1}^T(Ee^2_{\ell,t}e^2_{\ell,t'}-(Ee^2_{\ell,0})^2)\\
=O\left({N}{T}\right).
\end{align*}
Similarly,
\begin{align*}
&\sum_{\ell,\ell'=1, \ell\neq \ell'}^N\sum_{t,t'=1}^T(EX_{\ell,t}X_{\ell,t'}X_{\ell',t}X_{\ell',t'}-EX_{\ell,t}X_{\ell,t'}EX_{\ell',t}X_{\ell',t'})\\
&\hspace{.5cm}=\sum_{\ell,\ell'=1, \ell\neq \ell'}^N\gamma_\ell^2\gamma^2_{\ell'}\sum_{t,t'=1}^T(E\eta_t^2\eta^2_{t'}-1)
+2\sum_{\ell,\ell'=1, \ell\neq \ell'}^N\gamma_\ell^2\sum_{t,t'=1}^TE\eta_t\eta_{t'}Ee_{\ell',t}e_{\ell',t'}\\
&\hspace{1cm}
+\sum_{\ell,\ell'=1, \ell\neq \ell'}^N\sum_{t,t'=1}^TEe_{\ell,t}e_{\ell,t'} Ee_{\ell',t}e_{\ell',t'}\\
&=O(N^2T).
\end{align*}
We conclude from \eqref{dusch} that
\begin{align}\label{dusch-f}
\max_{1\leq i \leq K}\|\hfe_i-\hat{c}_i\fe_i\|=O_P\left({N}{T^{-1/2}}\right).
\end{align}
Next we define
$$
{v}^2_{i,T}=\sum_{s=-N+1}^{N-1}J\left(\frac{s}{h}\right){r}_{i,s},
$$
where
\begin{displaymath}
{r}_{i,s}=
\left\{
\begin{array}{ll}
\displaystyle \frac{1}{T-s}\sum_{t=1}^{T-s} {\xi}_{i,t}{\xi}_{i,t+s},\;\;\mbox{if}\;\;s\geq 0
\vspace{.2cm}\\
\displaystyle\frac{1}{T-|s|}\sum_{t=-s}^{T} {\xi}_{i,t}{\xi}_{i,t+s},\;\;\mbox{if}\;\;s< 0,
\end{array}
\right.
\end{displaymath}
where  ${\xi}_{i,t}=\fe_i^\T(\bX_t\bX_t^\T-E(\bX_0\bX_0^\T))\fe_i$.

We write
\begin{align*}
\tilde{v}^2_{i,T} - v^2_{i,T} = \sum_{j=1-N}^{N-1} J\left(\frac{j}{h}\right)(\tilde{r}_{j,s}-r_{j,s}).
\end{align*}

For $j \ge 0$,
\begin{align*}
\tilde{r}_{i,s}-r_{i,j}&= \frac{1}{T-j}\sum_{t=1}^{T-j} \tilde{\xi}_{i,t}\tilde{\xi}_{i,t+j} -  \xi_{i,t}\xi_{i,t+j} \\
&= \frac{1}{T-j}\sum_{t=1}^{T-j} (\tilde{\xi}_{i,t} -  \xi_{i,t})\tilde{\xi}_{i,t+j} + \frac{1}{T-j}\sum_{t=1}^{T-j} (\tilde{\xi}_{i,t+j} -  \xi_{i,t+j})\xi_{i,t}.
\end{align*}

According to the definitions of $\tilde{\xi}_{i,t}$ and $\xi_{i,t}$,

\begin{align*}
\tilde{\xi}_{i,t} - \xi_{i,t} =  (\hfe_i^\T - \fe_i^\T)U_t\hfe_i + \fe_i^\T U_t(\hfe_i - \fe_i),
\end{align*}
where $U_t=\bX_t \bX_t^\T - E\bX_0 \bX_0^\T$, from which it follows that,

\begin{align}
(\tilde{\xi}_{i,t} - \xi_{i,t})\tilde{\xi}_{i,t+j} =   (\hfe_i^\T - \fe_i^\T)U_t\hfe_i \hfe_i^\T U_{t+j} \hfe_i + \fe_i^\T U_t(\hfe_i - \fe_i) \hfe_i^\T U_{t+j} \hfe_i.
\end{align}
According to the Cauchy-Schwarz and triangle inequalities,
\begin{align*}
\left|\sum_{j=0}^{N-1} J\left(\frac{j}{h}\right)\frac{1}{T-j}\sum_{t=1}^{T-j} (\hfe_i^\T - \fe_i^\T)U_t\hfe_i \hfe_i^\T U_{t+j} \hfe_i \right| \le  \|\hfe_i - \fe_i\| \left|\sum_{j=0}^{N-1} J\left(\frac{j}{h}\right)\frac{1}{T-j}\sum_{t=1}^{T-j} \|U_t\| \|U_{t+j} \| \right|.
\end{align*}

According to \ref{dusch-f} $\|\hfe_i - \fe_i\|=O_P(NT^{-1/2})$. Furthermore, since $E\|U_0\|^2=O(N^2)$, and $J$ has bounded support, the Cauchy-Schwarz and triangle inequalities imply that

\begin{align}\label{lcalc-1}
E\left|\sum_{j=0}^{N-1} J\left(\frac{j}{h}\right)\frac{1}{T-j}\sum_{t=1}^{T-j} \|U_t\| \|U_{t+j} \| \right| \le c_1 h  E\|U_0\|^2= O(hN^2),
\end{align}

For some constant $c_1$. Hence, according to \eqref{lcalc-1} and Markov's inequality, we obtain that
\begin{align}
\left|\sum_{j=0}^{N-1} J\left(\frac{j}{h}\right)\frac{1}{T-j}\sum_{t=1}^{T-j} (\hfe_i^\T - \fe_i^\T)U_t\hfe_i \hfe_i^\T U_{t+j} \hfe_i \right| = O_P(hN^3T^{-1/2}).
\end{align}
Similar arguments applied to the remaining terms in $\tilde{v}^2_{i,T} - v^2_{i,T}$ show that
\begin{align}
|\tilde{v}^2_{i,T} - v^2_{i,T}|=O_P( hN^3T^{-1/2}).
\end{align}



It  follows from \eqref{model-null} and Assumption \ref{error-as}(b) that
\begin{align*}
\lim_{T\to\infty}\frac{1}{G(i,i)}\sum_{s=-\infty}^\infty E{\xi}_{i,t}{\xi}_{i,t+s}=1,
\end{align*}
and
\begin{align*}
\lim_{T\to\infty}\frac{1}{G(i,i)}\sum_{s=-\infty}^\infty K\left( \frac{s}{h}  \right)E {\xi}_{i,t}{\xi}_{i,t+s}=1.
\end{align*}
Since
\beq\label{var-0}
E{v}^2_{i,T}=\sum_{s=-\infty}^\infty K\left( \frac{s}{h}  \right)E {\xi}_{i,t}{\xi}_{i,t+s},
\eeq
if we show that
\beq\label{var-1}
\lim_{T\to\infty}\mbox{var}({v}^2_{i,T})=0,
\eeq
we get immediately that
$$
\frac{{v}^2_{i,T}}{G(i,i)}\;\;\;\stackrel{P}{\to}\;\;\;1,\quad{as  }\;\;\;T\to \infty.
$$
To this end, we have that
\begin{align*}
\mbox{var}({v}^2_{i,T})
=\sum_{s,s'=-\infty}^\infty J\left(\frac{s}{h}\right)J\left(\frac{s'}{h}\right)({r}_{i,s}-E\xi_{i,0}\xi_{i,s})({r}_{i,s'}-E\xi_{i,0}\xi_{i,s'})
\end{align*}
and
\begin{align*}
&\left|\sum_{s,s'=0}^\infty J\left(\frac{s}{h}\right)J\left(\frac{s'}{h}\right)({r}_{i,s}-E\xi_{i,0}\xi_{i,s})({r}_{i,s'}-E\xi_{i,0}\xi_{i,s'})\right|\\
&\hspace{.5cm}\leq \sum_{s,s'=0}^\infty \Biggl|J\left(\frac{s}{h}\right)J\left(\frac{s'}{h}\right)\Biggl|\frac{1}{T-s}\frac{1}{T-s'}\sum_{t=1}^{T-s}\sum_{t'=1}^{T-s'}\biggl|
E\fe_i^\T\X_t\X_t^\T\fe_i\fe_i^\T\X_{t+s}\X_{t+s}^\T\fe_i\fe_i^\T\X_{t'}\X_{t'}^\T\fe_i\\
&\hspace{.9cm}\times\fe_i^\T\X_{t'+s'}\X_{t'+s'}^\T\fe_i
-E\fe_i^\T\X_t\X_t^\T\fe_i\fe_i^\T\X_{t+s}\X_{t+s}^\T\fe_iE\fe_i^\T\X_{t'}\X_{t'}^\T\fe_i\fe_i^\T\X_{t'+s'}\X_{t'+s'}^\T\fe_i\biggl|\\
&\hspace{.5cm}\leq c_{7,7}\frac{1}{T^2} \sum_{s,s'=0}^h \sum_{t=1}^{T-s}\sum_{t'=1}^{T-s'}\biggl|
E\fe_i^\T\X_t\X_t^\T\fe_i\fe_i^\T\X_{t+s}\X_{t+s}^\T\fe_i\fe_i^\T\X_{t'}\X_{t'}^\T\fe_i\fe_i^\T\X_{t'+s'}\X_{t'+s'}^\T\fe_i\\
&\hspace{.9cm}
-E\fe_i^\T\X_t\X_t^\T\fe_i\fe_i^\T\X_{t+s}\X_{t+s}^\T\fe_iE\fe_i^\T\X_{t'}\X_{t'}^\T\fe_i\fe_i^\T\X_{t'+s'}\X_{t'+s'}^\T\fe_i\biggl|\\
&\hspace{.5cm}=O\left(\frac{h}{T}\right),
\end{align*}
with some constant $c_{7,7}$, since we can assume without loss of generality that $J(u)=0$ if $|u|\geq 1$.\\

Now we assume that the conditions of Theorem \ref{main-all-2} are satisfied. First we prove \eqref{cos-2}. It follows from \eqref{model-null} and \eqref{var-0} that
$$
\lim_{T\to \infty}\frac{1}{\|\bgamma\|^4}Er^2_{1,T}=V_1.
$$
Following the proof of one can verify that
$$
\mbox{var}\left( \frac{1}{\|\bgamma\|^4}r^2_{1,T} \right)=0,
$$
completing the proof of \eqref{cos-2}. The proof of \eqref{cos-3} goes along the lines of that of \eqref{cos-1} and therefore the details are omitted.\qed

\medskip
\noindent
{\it Proof of Theorem \ref{th-cons}.} We can assume without loss of generality that $\mu_i=0, 1\leq i \leq N$. Using \eqref{model-full} we have
\begin{displaymath}
\sum_{t=1}^{s}(\X_t-\bar{\X}_T)(\X_t-\bar{\X}_T)^\T=
\left\{
\begin{array}{ll}
\displaystyle s\left(\frac{T-t^*}{T}\right)^2\balpha\balpha^\T +\sum_{u=1}^s\U_{u,T},\;\; \mbox{if}\;\;\;0\leq s\leq t^*
\vspace{.2cm}\\
\displaystyle \left(t^*\left(\frac{T-t^*}{T}\right)^2+(s-t^*)\left(\frac{t^*}{T}\right)^2\right)\balpha\balpha^\T +\sum_{u=1}^s\U_{u,T},
\vspace{.2cm}\\
\hspace{3cm}\;\;\mbox{if}\;\;\; t^*\leq s \leq T,
\end{array}
\right.
\end{displaymath}
where
\begin{align*}
\U_{u,T}=&(\bgamma\eta_u+\be_u)(\bgamma\eta_u+\be_u)^\T-\frac{T-t^*}{T}(\bgamma\eta_u+\be_u)\balpha^\T-(\bgamma\eta_u+\be_u)\bZ_T^\T
-\frac{T-t^*}{T}\balpha(\bgamma\eta_u+\be_u)^\T\\
&+\frac{T-t^*}{T}\balpha\bZ_T^\T-\bZ_T(\bgamma\eta_u+\be_u)^\T+\frac{T-t^*}{T}\bZ_T\balpha^\T+\bZ_T\bZ_T^\T,\;\;\mbox{if}\;\;1\leq u \leq t^*,
\end{align*}
with
$$
\bZ_T=\bgamma\frac{1}{T}\sum_{v=1}^T\eta_v+\frac{1}{T}\sum_{v=1}^T\be_v, \quad\be_v=(e_{1,v}, e_{2,v},\ldots ,e_{N,v})^\T
$$
and
\begin{align*}
\U_{u,T}=&(\bgamma\eta_u+\be_u)(\bgamma\eta_u+\be_u)^\T+\frac{t^*}{T}(\bgamma\eta_u+\be_u)\balpha^\T-(\bgamma\eta_u+\be_u)\bZ_T^\T
+\frac{t^*}{T}\balpha(\bgamma\eta_u+\be_u)^\T\\
&-\frac{t^*}{T}\balpha\bZ_T^\T-\bZ_T(\bgamma\eta_u+\be_u)^\T-\frac{t^*}{T}\bZ_T\balpha^\T+\bZ_T\bZ_T^\T,\;\;\mbox{if}\;\;t^*\leq u \leq T.
\end{align*}
It follows along the same lines as the proof of \eqref{ineq**} that
$$
\sup_{0\leq s\leq T}\left|\fe ^\T\sum_{u=1}^s\U_{u,T}\fe\right|=O_P(T^{1/2})
$$
for $\fe\in R^N$ with $\|\fe\|=1$. Hence
$$
\frac{\hat{\lambda}_{1,T}(t^*/T)}{\|\balpha\|}\;\;\stackrel{P}{\to}\;\;\theta(1-\theta)^2
$$
and
$$
\frac{\hat{\lambda}_{1,T}(1)}{\|\balpha\|}\;\;\stackrel{P}{\to}\;\;\theta(1-\theta),
$$
which completes the proof of Theorem \ref{th-cons}.

\qed
\medskip

The proof of Theorem \ref{load-cons} is based on the following lemma.

\begin{lemma}\label{cons-load-lem} We assume that model \eqref{load-2} holds, Assumptions \ref{inc}, \ref{b-g}, \ref{error-as},  \eqref{n-t-2}, \eqref{asth-1} are satisfied.
If for some $0<\epsilon<1$ there exists an $N_0$ such that for all $N\geq N_0$
\beq\label{load-assa}
\left|\frac{\sup\left\{\bv^\T[ \theta\bgamma\bgamma^\T+(1-\theta)(\bgamma+\bdelta)(\bgamma+\bdelta)^\T+\bLambda  ]\bv:\;\bv\in R^N, \|\bv\|=1\right\}}{\sup\left\{
\bv^\T[\bgamma\bgamma^\T +\bLambda]\bv:\;\bv\in R^N, \|\bv\|=1\right\}} -1
\right|>\epsilon
\eeq
\end{lemma}
\begin{proof} Since the means of the panels do not change during the observation period in \eqref{load-2}, we can assume without loss of generality that $\mu_i=0, 1\leq i \leq N$. It follows from Theorems \ref{main-all} and \ref{main-all-2} that for any $u^*\in (0,\theta ]$  that
$$
\left|{\hat{\lambda}_1 (u^*)}-{\lambda_1}\right|=O_P\left(NT^{-1/2}\right).
$$
One can show that Lemmas \ref{rem-m} and \ref{eig-rem} hold with minor modifications under model \eqref{load-2}, and thus
$$
\left|{\hat{\lambda}_1 (1)}-{\bar{\lambda}_1(1)}\right|=O_P\left(NT^{-1}\right),
$$
where $\bar{\lambda}_1(1)$ is the largest eigenvalue of $\sum_{t=1}^T\X_t\X_t^\T/T$. Simple arithmetic shows that
$$
\frac{1}{T}\sum_{t=1}^T\X_t\X_t^\T=\C_T^{(1)}+\bG_{1,T}+\bG_{2,T},
$$
where
$$
\C_T^{(1)}=\bgamma\bgamma^\T\frac{1}{T}\sum_{t=1}^T\eta_t^2+\frac{1}{T}\sum_{t=1}^T\be_t\be_t^\T+(\bdelta\bdelta^\T +\bgamma\bdelta^\T +\bdelta\bgamma^\T)\frac{1}{T}\sum_{t=t^*}^T\eta_t^2,
$$
$$
\bG_{1,T}=\frac{1}{T}\sum_{t=1}^T(\bgamma\be_t^\T+\be_t\bgamma^\T)\eta_t
$$
and
$$
\bG_{2,T}=\frac{1}{T}\sum_{t=t^*}^T(\be_t\bdelta^\T+\bdelta\be_t^\T)\eta_t.
$$
It follows along the lines of the proof of \eqref{dusch-f} that
$$
\|\bG_{i,T}\|=O_P(NT^{-1/2}),\quad i=1,2,
$$
and thus if  $\lambda_T^{(1)}$ denotes the largest eigenvalue of $\C_T^{(1)}$, then we also have that
$$
\left|{\bar{\lambda}_1(1)}-{\lambda_T^{(1)}}  \right|=O_P(NT^{-1/2}).
$$
Let $\phi_T$ be the largest eigenvalue of $(\bdelta\bdelta^\T+\bdelta\bgamma^\T+\bgamma\bdelta^\T)(1-\theta)+\bgamma\bgamma^\T+\bLambda$. Then one can show using the arguments establishing Theorems \ref{main-all} and \ref{main-all-2} that
$$
\left| {\lambda_T^{(1)}}-{\phi_T} \right|=O_P(NT^{-1/2}).
$$
Assumption \eqref{load-assa} implies that there is an $\epsilon >0$ for all $T$ sufficiently large
$$
\left|\frac{\lambda_1}{\phi_T}-1  \right|>\epsilon,
$$
and therefore there is a constant $c_{8,1}$ such that
$$
\left|{\lambda_1}-{\phi_T}  \right|>c_{8,1}.
$$
Observing that $\hat{v}_{1,T}=O_P(h^{1/2})$ and $\sup_{0\leq u \leq 1}|\hat{B}_{T,1}(u)|\geq |\hat{B}_{T,1}(u^*)|$, the proof of \eqref{th-c-1} is complete.
\end{proof}

\medskip
\noindent
{\it Proof of Theorem \ref{load-cons}.} It is clear that assumption \eqref{load-3} implies \eqref{load-assa},  and therefore Theorem \ref{load-cons} follows from Lemma \ref{cons-load-lem}.

\end{document}